\PassOptionsToPackage{nameinlink}{cleveref}
\PassOptionsToPackage{pagebackref}{hyperref}
\documentclass[a4paper,UKenglish,numberwithinsect,cleveref,thm-restate]{lipics-v2021}
\pdfoutput=1 % for arXiv

\usepackage{chngcntr}

\usepackage{tikz}
\usetikzlibrary{arrows}
\usepackage[ruled]{algorithm2e}

\usepackage{mathtools} %\coloneqq
\usepackage{stmaryrd} %\llbracket

\colorlet{darkblue}{blue!50!black}
\colorlet{darkgreen}{green!50!black}
\hypersetup{
    breaklinks=true,
    colorlinks,
    citecolor=darkgreen,
    linkcolor=violet,
    urlcolor=darkblue
}

\hideLIPIcs
\nolinenumbers

\crefname{section}{Sect.}{Sections}
\Crefname{section}{Section}{Sections}

\newcommand{\Clustering}{\hyperref[thmClustering]{Clustering}}
\newcommand{\ClusteringLemma}{\hyperref[thmClustering]{Clustering Lemma}}
\newcommand{\Separation}{\hyperref[thmSeparation]{Separation}}
\newcommand{\SeparationLemma}{\hyperref[thmSeparation]{Separation Lemma}}

\newcommand{\AbstractionLemma}{\hyperref[thmAbstraction]{Abstraction Lemma}}
\newcommand{\SymmetryLemma}{\hyperref[thmSymmetry]{Symmetry Lemma}}

\theoremstyle{claimstyle}
\newtheorem{localclaim}{Claim}
\newcommand{\resetlocalclaim}{\setcounter{localclaim}{0}}

\newenvironment{Example}{\example}{\lipicsEnd\endexample}

\newcommand{\enumref}[1]{\textcolor{lipicsGray}{\sffamily\bfseries\upshape\mathversion{bold}(#1)}}

\newenvironment{itemise}{\smallskip\begin{itemize}\setlength{\itemsep}{0.25em}}{\end{itemize}\smallskip}
\newenvironment{bracketise}{\smallskip\begin{bracketenumerate}\setlength{\itemsep}{0.25em}}{\end{bracketenumerate}\smallskip}
\newenvironment{enumise}{\smallskip\begin{enumerate}\setlength{\itemsep}{0.25em}}{\end{enumerate}\smallskip}
\newenvironment{alphatise}{\smallskip\begin{alphaenumerate}\setlength{\itemsep}{0.25em}}{\end{alphaenumerate}\smallskip}

\newcommand{\ra}{\rightarrow}
\newcommand{\E}{\exists}
\newcommand{\A}{\forall}
\newcommand{\biglor}{\bigvee}
\newcommand{\bigland}{\bigwedge}

\renewcommand{\phi}{\varphi}
\renewcommand{\theta}{\vartheta}
\renewcommand{\emptyset}{\varnothing}
\renewcommand{\epsilon}{\varepsilon}
\renewcommand{\AA}{{\mathfrak A}}

\newcommand{\N}{{\mathbb N}}

\DeclareMathOperator{\arity}{arity}

\newcommand{\FO}{{\rm FO}}

\DeclareMathOperator{\Lit}{\mathrm{Lit}}

\newcommand*{\eps}{\epsilon}
\newcommand*{\dcup}{\mathbin{\dot{\cup}}}
\newcommand*{\from}{\colon}

\newcommand*{\tup}[1]{\bar{#1}}
\renewcommand{\bar}{\overline}
\newcommand{\ta}{\tup a}
\newcommand{\tb}{\tup b}
\newcommand{\tc}{\tup c}
\newcommand{\td}{\tup d}
\newcommand{\tx}{\tup x}
\newcommand{\ty}{\tup y}
\newcommand{\tz}{\tup z}
\newcommand{\tu}{\tup u}
\newcommand{\tv}{\tup v}

\newcommand{\Bool}{\mathbb{B}}

\renewcommand{\N}{{\mathbb N}}

\newcommand{\Trop}{\mathbb{T}}
\newcommand{\Vit}{\mathbb{V}}

\newcommand*{\Real}{\mathbb{R}}

\newcommand*{\ext}[1]{[\![ #1 ]\!]}
\newcommand*{\Ext}[1]{\big\llbracket #1 \big\rrbracket}
\newcommand*{\EXT}[1]{\Big\llbracket #1 \Big\rrbracket}

\newcommand{\meet}{\sqcap}
\newcommand{\join}{\sqcup}

\newcommand{\keq}[1][K]{\equiv_{#1}}

\newcommand{\Minmax}{\mathcal M}
\newcommand{\Lattice}{\mathcal L}

\newcommand{\meq}{\equiv_\Minmax}
\newcommand{\mle}{\le_\Minmax}
\newcommand{\mge}{\ge_\Minmax}
\newcommand{\meqtag}[1]{\,\overset{\mathclap{(#1)}}{\equiv}_\Minmax\,}

\newcommand{\ball}[2]{B_{#1}(#2)} % \ball r \ta
\newcommand{\ballpi}[3][\pi]{B_{#2}^{#1}(#3)} % \ballpi r \ta
\newcommand{\balltau}[3][\tau]{B_{#2}^{#1}(#3)}

\newcommand{\Qball}[4][Q]{#1 #2 {\in} \ball{#3}{#4}}
\newcommand{\Eball}[3]{\Qball[\E]{#1}{#2}{#3}} % \Aball z r y (forall z in B_r(y))
\newcommand{\Aball}[3]{\Qball[\A]{#1}{#2}{#3}} % \Eball z r y (exists z in B_r(y))

\newcommand{\Qballtau}[5][Q]{#1 #3 {\in} \balltau[#2]{#4}{#5}}
\newcommand{\Eballtau}[4]{\Qballtau[\E]{#1}{#2}{#3}{#4}}
\newcommand{\Aballtau}[4]{\Qballtau[\A]{#1}{#2}{#3}{#4}}

\newcommand{\Qnotball}[4][Q]{#1 #2 {\notin} \ball{#3}{#4}}

\newcommand{\Anotball}[3]{\Qnotball[\A]{#1}{#2}{#3}}

\newcommand{\Qsc}[3][Q]{{#1}^{#2\textrm{-sc}}(#3)}
\newcommand{\Esc}[2]{\Qsc[\E]{#1}{#2}} % \exsc {2r} {x_1,\dots,x_n}
\newcommand{\Asc}[2]{\Qsc[\A]{#1}{#2}}
\newcommand{\Qdistinct}[2][Q]{{#1}^{\text{distinct}}(#2)}
\newcommand{\Edistinct}[1]{\Qdistinct[\E]{#1}} % \exdistinct {x_1,\dots,x_n}
\newcommand{\Adistinct}[1]{\Qdistinct[\A]{#1}}

\newcommand{\loc}[1]{^{(#1)}} % \phi \loc r
\newcommand{\Part}{\text{Part}} % partitions
\newcommand{\calS}{{\mathcal S}} % class of semirings
\newcommand{\dual}[1]{{#1}^*} % symmetry lemma

\title{Locality Theorems in Semiring Semantics}

\author{Clotilde Bizi{\` e}re}{ENS Paris, France}{clotilde.biziere@ens.psl.eu}{}{}
\author{Erich Gr{\"a}del}{RWTH Aachen University, Germany}{graedel@logic.rwth-aachen.de}{https://orcid.org/0000-0002-8950-9991}{}
\author{Matthias Naaf}{RWTH Aachen University, Germany}{naaf@logic.rwth-aachen.de}{https://orcid.org/0000-0002-1099-5713}{}

\authorrunning{C. Bizi{\` e}re, E. Gr{\"a}del, and M. Naaf}

\Copyright{Clotilde Bizi{\` e}re, Erich Gr{\"a}del, and Matthias Naaf}


\relatedversion{A conference version of this paper appears at MFCS'23}

\acknowledgements{We thank Steffen~van~Bergerem and Max~Wehmeier for pointing out errors in earlier versions of this paper.}

\begin{document}

\maketitle

%%%%%%%%%%%%%%%%%%%%%%%%%%%%%%%%
%%  METADATA
%%%%%%%%%%%%%%%%%%%%%%%%%%%%%%%%

\begin{abstract}
Semiring semantics of first-order logic generalises classical Boolean semantics by permitting truth values from a commutative semiring,
which can model information such as costs or access restrictions.
This raises the question to what extent classical model-theoretic properties still apply, and how this depends on the algebraic properties of the semiring.

In this paper, we study this question for the classical locality theorems due to Hanf and Gaifman.
We prove that Hanf's locality theorem generalises to all semirings with idempotent operations, but fails for many non-idempotent semirings.
We then consider Gaifman normal forms and show that for formulae with free variables, Gaifman's theorem does not generalise beyond the Boolean semiring.
Also for sentences, it fails in the natural semiring and the tropical semiring.
Our main result, however, is a constructive proof of the existence of Gaifman normal forms for min-max and lattice semirings.
The proof implies a stronger version of Gaifman's classical theorem in Boolean semantics:
every sentence has a Gaifman normal form which does not add negations.
\end{abstract}

%%%%%%%%%%%%%%%%%%%%%%%%%%%%%%%%
%%  MAIN PART
%%%%%%%%%%%%%%%%%%%%%%%%%%%%%%%%

\section{Introduction}

Originally motivated by \emph{provenance analysis in databases} (see e.g. \cite{GreenTan17,Glavic21} for surveys), 
semiring semantics is based on the idea to evaluate logical statements not just by \emph{true}
or \emph{false}, but by values in some commutative semiring $(K,+,\cdot,0,1)$.
In this context, the standard semantics appears as the special case 
when the Boolean semiring $\Bool = (\{\bot, \top\}, \lor, \land, \bot, \top)$ is used.
Valuations in other semirings provide additional information, beyond truth
or falsity: the tropical semiring $\Trop= (\Real_{+}^{\infty}, \min, +, \infty, 0)$
is used for \emph{cost analysis}, the natural semiring $\N=(\N,+,\cdot,0,1)$ for counting
evaluation strategies and proofs,
and the Viterbi-semiring $\Vit = ([0,1]_{\Real}, \max, \cdot, 0, 1)$ models \emph{confidence scores}.
Finite or infinite min-max semirings $(K, \max, \min, a, b)$
can model, for instance, different \emph{access levels}
to atomic data (see e.g.\ \cite{FosterGreTan08}); valuations of a first-order sentence $\psi$ in such \emph{security semirings} determine the 
required clearance level that is necessary to access enough information
to determine the truth of $\psi$.
Further, semirings of polynomials or formal power series  permit us to \emph{track} which
atomic facts are used (and how often) to establish the truth of a sentence in a given structure, and
this has applications for database repairs \cite{XuZhaAlaTan18}
and also for the strategy analysis of games \cite{GraedelTan20,GraedelLueNaa21}. 
Semiring semantics replaces structures
by \emph{$K$-interpretations}, which are functions $\pi\from \Lit_A(\tau)\ra K$, mapping fully instantiated
$\tau$-literals $\phi(\bar a)$ over a universe $A$ to values in a commutative semiring $K$. 
The value $0\in K$ is interpreted as \emph{false}, while all other values in $K$  are viewed as
nuances of \emph{true} or, perhaps more accurately, as \emph{true, with some additional information}.
In provenance analysis, this is sometimes referred to as \emph{annotated facts}.
The value $1\in K$ is used to represent \emph{untracked information}
and is used in particular to evaluate true equalities and inequalities.

The development of semiring semantics raises the question to what extent classical techniques and results of logic extend to semiring semantics, 
and how this depends on the algebraic properties of the underlying semirings.  
Previous investigations in this direction have studied, for instance, the relationship between elementary equivalence and isomorphism
for finite semiring interpretations and their definability up to isomorphism \cite{GraedelMrk21}, 
Ehrenfeucht-Fra\"{\i}ss{\'e} games \cite{Brinke23}, and 0-1 laws \cite{GraedelHelNaaWil22}. 

The purpose of this paper is to study \emph{locality} in semiring semantics.
Locality is a fundamental property of first-order logic in classical semantics
and an important limitation of its expressive power.  It
means that the truth of a first-order formula $\psi(\bar x)$ in a given structure only depends on a 
neighbourhood of bounded radius around $\bar x$, and on the existence of a bounded number of local substructures.
Consequently, first-order logic cannot express global properties such as connectivity or acyclicity of graphs.
On graphs there are natural and canonical notions of the distance between two points and of a neighbourhood of a given radius around a point. 
To define these notions for an arbitrary relational structure $\AA$ one  associates with it its \emph{Gaifman graph} $G(\AA)=(A,E)$
where two points $a\neq b$ are adjacent if, and only if, they
coexist in some atomic fact. There exist several notions of locality; the most common ones are 
\emph{Hanf locality} and \emph{Gaifman locality}, and the fundamental locality theorems for first-order logic are
\emph{Hanf's locality theorem} and \emph{Gaifman's normal form theorem}.
In a nutshell, Hanf's theorem gives a criterion for the $m$-equivalence 
(i.e. indistinguishability by sentences of quantifier rank up to $m$) of two structures  
based on the number of  local substructures of any given isomorphism type, while
Gaifman's theorem
states that every first-order formula is equivalent to a Boolean combination of local formulae and basic local sentences,
which has many model-theoretic and algorithmic consequences. 
We shall present precise statements of these results in \cref{secHanf} and \cref{secGaifmanDefinitions}.

Locality thus provides powerful techniques,
also for logics that go beyond first-order logic by counting properties, generalised quantifiers, or aggregate functions, \cite{ArenasBarLib08,HellaLibNur99,HellaLibNurWon01,KuskeSch18}. It has 
applications in different areas including  low-complexity model-checking algorithms \cite{GroheKreSie17,GroheWoe04}, 
approximation schemes for logically defined optimisation problems \cite{DawarGroKreSch06},
automata theory \cite{SchwentickBar98}, 
computational issues on database transactions \cite{BenediktGriLib98}, and most recently also in learning theory,
for the efficient learning of logical concepts
\cite{Bergerem19,BergeremSch21,BergeremGroRit22}.
This motivates the question, whether locality is also applicable in semiring semantics.
The relevant semiring interpretations in this context are 
\emph{model-defining}, which means that  for any pair of complementary literals $R\bar a, \neg R\bar a$
precisely one of the values $\pi(R\bar a)$, $\pi(\neg R\bar a)$ is 0,
and \emph{track only positive information} which means that
$\pi(\neg R\bar a)$ can only take the values 0 or 1.
Model defining interpretations $\pi$ define a unique structure $\AA_\pi$ and we thus obtain
a well-defined Gaifman graph $G(\pi) \coloneqq G(\AA_\pi)$, with the associated notions
of distance and neighbourhoods. The assumption that only positive information is tracked
is necessary to get meaningful locality properties (see \cref{sec:semiringsemantics}).

\medskip We clearly cannot generalise all known locality properties of
first-order logic to semiring semantics in arbitrary commutative semirings. On semirings whose operations 
are not idempotent,  we cannot expect a Gaifman normal form, since for computing the value of a quantified statement, 
we have to add or multiply values of subformulae for
\emph{all} elements of the structure, which gives an inherent source of non-locality.  As a consequence, some of the
locality results that we prove hold only under certain algebraic assumptions on the semiring, and further there
turns out to be a difference of the locality properties of sentences and those of formulae with free variables.
We shall establish the following results.
\begin{bracketenumerate}
\item First-order formulae are Hanf-local for all semirings.
\item Hanf's locality theorem generalises to all fully idempotent semirings (in which  both addition and multiplication are idempotent).
\item For formulae with free variables, Gaifman's normal form theorem does not generalise beyond the Boolean semiring.
\item For sentences, Gaifman's normal form theorem also fails in certain important semirings such as the natural semiring and the tropical semiring.
\item Over min-max semirings (and even lattice semirings), every first-order sentence has a Gaifman normal form.
\item In classical Boolean semantics, every sentence has a Gaifman normal form which does not introduce new negations.
\end{bracketenumerate}

The results \enumref{1}, \enumref{2} on Hanf locality (\cref{secHanf}) are proved by adaptations of the arguments for the Boolean case.
The results \enumref{3} and \enumref{4} are established in \cref{secCounterexamples} via specific examples of formulae that defeat locality, using simple algebraic arguments.
The most ambitious result and the core of our paper is \enumref{5}, a version of Gaifman's theorem for min-max semirings (\cref{secGaifmanProof}), which we later generalise to lattice semirings (\cref{secStrengthening}).
It requires a careful choice of the
right syntactical definitions for local sentences and, since the classical proofs in \cite{Gaifman82,EbbinghausFlu99} do not seem to generalise to semiring semantics, a new approach for the proof, based on quantifier elimination.
This new approach also leads to a stronger version of Gaifman's theorem in Boolean semantics \enumref{6}, which might be of independent interest.

\section{Semiring Semantics}\label{sec:semiringsemantics}

This section gives a brief overview on semiring semantics of first-order logic (see \cite{GraedelTan17} for more details) and the relevant algebraic properties of semirings.
We further define generalisations of the classical notions of equivalence, isomorphism and distance (via the Gaifman graph).

A commutative%
\footnote{In the following, \emph{semiring} always refers to a commutative semiring.}
semiring is an algebraic structure $(K,+,\cdot,0,1)$ with $0\neq1$, such that $(K,+,0)$
and $(K,\cdot,1)$ are commutative monoids, $\cdot$
distributes over $+$, and $0\cdot a=a\cdot 0=0$. We focus on semirings that are \emph{naturally ordered},
in the sense that  $a\leq b:\Leftrightarrow \E c (a+c=b)$ is a partial order.
For the study of locality properties, an important subclass are the \emph{fully idempotent} semirings, in which both operations are idempotent (i.e., $a+a=a$ and $a \cdot a = a$).
Among these, we consider in particular all \emph{min-max} semirings $(K,\max,\min,0,1)$ induced by a total order $(K,\le)$ with minimal element $0$ and maximal element $1$,
and the more general \emph{lattice} semirings $(K,\join,\meet,0,1)$ induced by a bounded distributive lattice $(K,\le)$.

For a finite relational vocabulary  $\tau$ and a finite universe $A$, we write $\Lit_A(\tau)$ for the set of
\emph{instantiated} $\tau$-literals $R\ta$ and $\neg R\ta$ with $\ta \in A^{\arity(R)}$.
Given a commutative semiring $K$, a \emph{$K$-interpretation} (of vocabulary $\tau$ and universe $A$)
is a function $\pi\from\Lit_A(\tau)\to K$.
It is \emph{model-defining} if for any pair of complementary literals $L$, $\neg L$
precisely one of the values $\pi(L)$, $\pi(\neg L)$ is $0$.
In this case, $\pi$ induces a unique (Boolean) $\tau$-structure $\AA_\pi$ with universe $A$ such that,
for every literal $L\in\Lit_A(\tau)$, we have that $\AA_\pi\models L$ if, and only if, 
$\pi(L)\neq 0$.

A $K$-interpretation $\pi \from \Lit_A(\tau) \to K$ extends in a straightforward way
to a valuation $\pi \ext{\phi(\ta)}$
of any instantiation of a formula $\phi(\tx) \in \FO(\tau)$, 
assumed to be written in negation normal form,
by a tuple $\ta \subseteq A$.
The semiring semantics $\pi \ext{\phi(\ta)}$ is defined
by induction. We first extend $\pi$ by mapping equalities and inequalities to their truth values, by 
setting $\pi \ext{a = a} \coloneqq 1$ and $\pi\ext{a=b} \coloneqq 0$ for $a\neq b$ (and analogously for
inequalities).
Further, disjunctions and existential quantifiers are interpreted as sums, and conjunctions and universal quantifiers as products:
\begin{alignat*}{3}
\pi \ext{\psi(\ta) \lor \theta(\ta)} &\coloneqq \pi \ext{\psi(\ta)} + \pi \ext{\theta(\ta)} &\quad\quad\quad \pi \ext{\psi(\ta) \land \theta(\ta)} &\coloneqq \pi \ext{\psi(\ta)} \cdot \pi \ext{\theta(\ta)} \\
\pi \ext{\exists x \, \theta(\ta, x)} &\coloneqq \sum_{a \in A} \pi \ext{\theta(\ta, a)} &\quad\quad\quad \pi \ext{\forall x \, \theta(\ta, x)} &\coloneqq \prod_{a \in A} \pi \ext{\theta(\ta, a)}.
\end{alignat*}
Since negation does not correspond to a semiring operation, we insist on writing all formulae in negation normal form.
This is a standard approach in semiring semantics (cf.\ \cite{GraedelTan17}).

Equivalence of formulae now takes into account the semiring values and is thus more fine-grained than Boolean equivalence.
We often consider equivalence transformations that hold for an entire class of semirings, such as all min-max semirings.

\begin{definition}[$\keq$]
Two formulae $\psi(\tx)$, $\phi(\tx)$ are
$K$-equivalent (denoted $\psi\keq\phi$) if $\pi\ext{\psi(\ta)}=\pi\ext{\phi(\ta)}$ for every model-defining
$K$-interpretation $\pi$ (over finite universe) and every tuple $\ta$.

For a class $\calS$ of semirings, we write  $\psi\equiv_{\calS}\phi$ if
$\psi\keq\phi$ holds for all $K\in\calS$.
\end{definition}

\medskip
Basic mathematical notions such as isomorphisms, partial isomorphisms and elementary equivalence 
naturally extend from relational structures to $K$-interpretations.
We lift bijections $\sigma \from A \to B$ to literals $L \in \Lit_A(\tau)$ in the obvious way, i.e., $\sigma(L) \in \Lit_B(\tau)$ results from $L$ by replacing each element $a \in A$ with $\sigma(a) \in B$.

\begin{definition}[Isomorphism]
Let $\pi_A \from \Lit_A(\tau) \to K$ and $\pi_B \from \Lit_B(\tau) \allowbreak \to K$ be two $K$-interpretations.
We say that $\pi_A$ and $\pi_B$ are \emph{isomorphic} (denoted $\pi_A \cong \pi_B$) if there is a bijection $\sigma \from A \to B$ such that
$\pi_A(L) = \pi_B(\sigma(L))$ for all  $L \in \Lit_A(\tau)$.

A \emph{partial isomorphism} between $\pi_A$ and $\pi_B$ is a bijection $\sigma \from X \to Y$ on subsets
$X \subseteq A$ and $Y \subseteq B$ such that $\pi_A(L)=\pi_B(\sigma(L))$ for all literals $L \in L_X(\tau)$
(i.e., literals that are instantiated
with elements from $X$ only).
\end{definition}

\begin{definition}[Elementary equivalence]
Let $\pi_A \from \Lit_A(\tau) \to K$ and $\pi_B \from \Lit_B(\tau) \to K$ be two $K$-interpretations,
and $\ta \in A^r$ and $\tb \in B^r$ be tuples of the same length.
The pairs $\pi_A, \ta$ and $\pi_B, \tb$ are \emph{elementarily equivalent}, denoted $\pi_A, \ta \equiv \pi_B, \tb$, if
$\pi_A \ext{\phi(\ta)} = \pi_B \ext{\phi(\tb)} $ for all  $\phi(\tx) \in \FO(\tau)$.
They are $m$-\emph{equivalent}, denoted  $\pi_A, \ta \equiv_m \pi_B, \tb$,
if the above holds for all $\phi(\tx)$ with quantifier rank at most $m$.
\end{definition}

It is obvious that, as in classical semantics,
isomorphism implies elementary equivalence.
For finite universes, the converse is true in classical semantics,
but fails in semiring semantics for certain semirings, including very simple ones such as finite min-max semirings
(see \cite{GraedelMrk21}).

\medskip
Towards locality properties, we define distances between two elements $a,b$ in a $K$-interpretation $\pi$
based on the induced structure $\AA_\pi$.

\begin{definition}[Gaifman graph]
The \emph{Gaifman graph} $G(\pi)$ of a model-defining $K$-interpre\-ta\-tion $\pi \from \Lit_A(\tau) \to K$
is defined as the Gaifman graph $G(\AA_\pi)$ of the induced $\tau$-structure.
That is, two elements $a\neq b$ of $A$ are adjacent in $G(\AA_\pi)$ if, and only
if, there exists a positive literal $L = R c_1 \dots c_r \in \Lit_A(\tau)$ such that $\pi(L)\neq 0$ and
$a,b\in\{c_1,\dots c_r\}$.

We write $d(a,b) \in \N$ for the distance of $a$ and $b$ in $G(\pi)$.
We further define the \emph{$r$-neighbourhood} of an element $a$ in $\pi$ as $\ballpi r a \coloneqq \{ b\in A: d(a,b) \leq r\}$.
For a tuple $\ta\in A^k$ we put $\ballpi r \ta \coloneqq \bigcup_{i\leq k} \ballpi r {a_i}$.
\end{definition}

Locality properties are really meaningful only for
semiring interpretations $\pi\from \Lit_A(\tau)\ra K$ that \emph{track only positive information},
which means that $\pi(\neg L)\in\{0,1\}$ for each negative literal $\neg L$. Indeed, if also negative literals carry 
non-trivial information, then either these must be taken into account in the definition of 
what ``local'' means, which will trivialise the Gaifman graph (making it a clique) 
so locality would become  meaningless, or otherwise local information no longer suffices to determine values of even very simple 
sentences involving negative literals, such as $\E x\E y\neg Rxy$.  
We therefore consider here only $K$-interpretations over finite universes which are model-defining and track only positive information.

\section{Hanf Locality}
\label{secHanf}

The first formalisation of locality we consider is Hanf locality.
We present generalisations of both the Hanf locality rank and of Hanf's locality theorem, where the latter is conditional on algebraic properties of the semirings (cf.\ \cite{EbbinghausFlu99,Libkin04} for the classical proofs).

\subsection{Hanf Locality Rank}

Recall that in classical semantics, every first-order formula is Hanf-local with locality rank depending only on the quantifier rank.
By a straightforward adaptation of the classical proof, it turns out that also in semiring semantics, every first-order formula is Hanf-local. 

\begin{definition}[$\rightleftharpoons_r$]
Let $\pi_A \from \Lit_A(\tau) \to K$ and $\pi_B \from \Lit_B(\tau) \to K$ be two $K$-interpre\-ta\-tions.
For tuples $\ta \subseteq A$ and $\tb \subseteq B$ of matching length, we write $(\pi_A,\ta) \rightleftharpoons_r (\pi_B, \tb)$ if there is a bijection $f\from A \to B$ such that $B_r^{\pi_A}(\ta,c) \cong B_r^{\pi_B}(\tb,f(c))$ for all $c \in A$. 
\end{definition}

\begin{proposition}[Hanf locality in semiring semantics] \label{thmHanfLocal}
Let $K$ be an arbitrary semiring. For every first-order formula $\phi(\tx)$, 
there exists $r\in\N$, depending only on the quantifier rank of $\phi$, such that for
all model-defining $K$-interpretations $\pi_A,\pi_B$ that track only positive information,
and all tuples $\ta$, $\tb$ we have that $\pi_A\ext{\phi(\ta)}=\pi_B\ext{\phi(\tb)}$ whenever
$(\pi_A,\ta) \rightleftharpoons_r (\pi_B, \tb)$.
\end{proposition}

This follows by a simple adaptation of the classical proof in \cite{Libkin04},
which relies on the inductive argument that whenever  
$(\pi_A, \ta) \rightleftharpoons_{3r+1} (\pi_B, \tb)$, then there exists a bijection $f\from A \to B$ 
such that  $(\pi_A, \ta,c) \rightleftharpoons_{r} (\pi_B, \tb, f(c))$ for all $c \in A$.
The only point that requires care is the combination of partial isomorphisms on disjoint and non-adjacent neighbourhoods,
which in our setting depends on the assumption that the $K$-interpretations only track positive information:

\begin{lemma} \label{lem-topi}
Let $\pi_A$ and $\pi_B$ be model-defining $K$-interpretations that track only positive information.
Let $\sigma \from B^{\pi_A}_r(\ta)\ra B^{\pi_B}_r(\tb)$ and $\sigma'\from B^{\pi_A}_r(\ta')\ra B^{\pi_B}_r(\tb')$
be two partial isomorphisms between disjoint $r$-neighbourhoods in $\pi_A$ and $\pi_B$.
If $d(\ta,\ta')>2r+1$ and $d(\tb,\tb')>2r+1$,
then $(\sigma\cup\sigma') \from B^{\pi_A}_r(\ta, \ta')\ra B^{\pi_B}_r(\tb, \tb')$ is also a partial isomorphism.
\end{lemma}

\begin{proof} Clearly $(\sigma\cup\sigma')$ is a bijection, so we only have to show that
$\pi_A(L)=\pi_B((\sigma\cup\sigma')(L))$ for every literal $L\in \Lit_X(\tau)$,
where $X=B^{\pi_A}_r(\ta\ta')$. For every literal $L$ instantiated with only elements
from $B^{\pi_A}_r(\ta)$ or only elements from $B^{\pi_A}_r(\ta')$ this is clear,
since $\sigma$ and $\sigma'$ are partial isomorphisms. So consider a literal
$L$ that is instantiated by elements from both  $B^{\pi_A}_r(\ta)$ and $B^{\pi_A}_r(\ta')$,
which implies that $L' \coloneqq (\sigma\cup\sigma')(L)$ is instantiated by
elements from both $B^{\pi_B}_r(\tb)$ and $B^{\pi_A}_r(\tb')$.
If $L$ is a positive literal, then $\pi_A(L)=0$ and $\pi_B(L')=0$,
as otherwise $d(\ta,\ta')\leq 2r+1$ or $d(\tb,\tb')\leq 2r+1$.
If $L$ is a negative literal, i.e., $L = \neg \hat L$ for a positive literal $\hat L$,
then by the same argument, $\pi_A(\hat L)=\pi_B(\hat L')=0$.
Since $\pi_A$ and $\pi_B$ are model-defining and track only positive information,
we have $\pi_A(L)=\pi_B(L')=1$.
\end{proof}

\subsection{Hanf's Locality Theorem}

Hanf's locality theorem provides a sufficient combinatorial criterion for the $m$-equivalence of two structures, i.e. for 
their indistinguishability by sentences of quantifier rank up to $m$.
We now turn to the question under what conditions this theorem generalises to 
semiring semantics.
Its classical proof (cf.\ \cite{EbbinghausFlu99}) proceeds
by showing that Hanf's criterion admits the construction of a back-and-forth system
$(I_j)_{j\leq m}$ which, by the Ehrenfeucht-Fra\"{\i}ss\'e theorem, implies the $m$-equivalence of the
two structures. It turns out that this method carries over to $K$-interpretations precisely in the
case that the semiring $K$ is fully idempotent. We further show that for semirings that are not fully idempotent,
there actually are counterexamples to Hanf's locality theorem.

\begin{definition}[Back-and-forth system]
Let $\pi_A$ and $\pi_B$ be two $K$-interpretations and let $k \geq 0$.
A \emph{$m$-back-and-forth system} for  $\pi_A$ and 
$\pi_B$ is a sequence $(I_j)_{j \leq m}$ of finite sets of partial 
isomorphisms between $\pi_A$ and $\pi_B$ such that
\begin{itemise}
\item
$\emptyset \in I_m$, and
\item
for all $j < m$, the set $I_{j+1}$ has back-and-forth extensions in $I_{j}$,
i.e., whenever $\ta\mapsto\tb \in I_{j+1}$ then for every $c\in A$ there exists $d\in B$,
and vice versa, such that $(\ta c)\mapsto (\tb d)$ is in  $I_j$.
\end{itemise}
We write $(I_j)_{j \leq m} \colon \pi_A \cong_m \pi_B$
if  $(I_j)_{j \leq m}$ is a $m$-back-and-forth system for $\pi_A$ and $\pi_B$. 
\end{definition}

Back-and-forth systems can be seen as algebraic descriptions of winning strategies in Ehrenfeucht-Fra\"{\i}ss\'e
games, and in classical semantics, an $m$-back-and-forth system between two structures exists
if, and only if, the structures are $m$-equivalent. However, in semiring semantics this equivalence may, in general,
fail in both directions \cite{Brinke23}. A detailed investigation of the relationship between elementary equivalence,
Ehrenfeucht-Fra\"{\i}ss\'e games, and
back-and-forth-systems in semiring semantics is outside the scope of this paper, and will be presented
in forthcoming work. For the purpose of studying Hanf locality, we shall need just the fact that in the
specific case of fully idempotent semirings, $m$-back-and-forth systems do indeed provide
a sufficient criterion for $m$-equivalence.

\begin{proposition}\label{prop-back-and-forth}
Let $\pi_A$ and $\pi_B$ be $K$-interpretations into a fully idempotent semiring $K$.
If there is an $m$-back-and-forth system $(I_j)_{j\leq m}$ for $\pi_A$ and 
$\pi_B$, then $\pi_A\equiv_m \pi_B$.
\end{proposition}  

\begin{proof} We show by induction that for every first-order formula $\psi(\tx)$
of quantifier rank $j\leq m$ and every partial isomorphism $\ta\mapsto\tb\in I_{j}$ we have that
$\pi_A\ext{\psi(\ta)}=\pi_B\ext{\psi(\tb)}$. For $j=0$ this is trivial.
For the inductive case it suffices to consider formulae $\psi(\tx)=\E y \, \phi(\tx,y)$
and $\psi(\tx)=\A y \, \phi(\tx,y)$, and a map $\ta\mapsto\tb\in I_{j+1}$.
We have that
\begin{align*}
    \pi_A\ext{\E y \,\phi(\ta,y)} &=\sum_{c\in A}\pi_A\ext{\phi(\ta,c)} &\text{and }\qquad 
    \pi_B\ext{\E y \,\psi(\tb,y)} &=\sum_{d\in B}\pi_B\ext{\phi(\tb,d)},\\
    \pi_A\ext{\A y \,\phi(\ta,y)} &=\prod_{c\in A}\pi_A\ext{\phi(\ta,c)} &\text{and }\qquad 
    \pi_B\ext{\A y \,\psi(\tb,y)} &=\prod_{d\in B}\pi_B\ext{\phi(\tb,d)}.
\end{align*}

Since the semiring is fully idempotent, the valuations $\pi_A\ext{\E y\,\phi(\ta,y)}$
and $\pi_A\ext{\A y\,\phi(\ta,y)}$
only depend on the \emph{set} of all values $\pi_A\ext{\phi(\ta,c)}$ for $c\in A$,
and not on their multiplicities.
It thus suffices to prove that the sets of
values are identical for $(\pi_A,\ta)$ and $(\pi_B,\tb)$, i.e.
\[     \{ \pi_A\ext{\phi(\ta,c)}: c\in A\}  =   \{ \pi_B\ext{\phi(\tb,d)}: d\in B\}.\]
But this follows immediately from the fact that $\ta\mapsto\tb$ has
back and forth extensions in $I_j$, and from the induction hypothesis:
for each $c\in A$ there exists some $d\in B$, and vice versa, such that the map
$(\ta,c)\mapsto(\tb,d)$ is in $I_j$, and therefore $\pi_A\ext{\phi(\ta,c)}=\pi_B\ext{\phi(\tb,d)}$.
\end{proof}

To formulate Hanf's criterion for $K$-interpretations $\pi_A,\pi_B$, we write
$\pi_A \rightleftharpoons_{r,t} \pi_B$, for $r,t\in\N$, 
if for every isomorphism type $\iota$ of $r$-neighbourhoods, either $\pi_A$ and $\pi_B$ have
the same number of realisations of $\iota$, or both have at least $t$ realisations.

\begin{theorem}[Hanf's theorem for fully idempotent semirings]
Let $K$ be a fully idempo\-tent semiring. For all $m,\ell\in \N$ there exist
$r=r(m) \in \N$ and $t=t(m,\ell)\in\N$ such that for all model-defining $K$-interpretations $\pi_A$ and $\pi_B$ that track only positive information  
and whose Gaifman graphs have maximal degree $\le \ell$, we have that $\pi_A\equiv_m \pi_B$ whenever
$\pi_A \rightleftharpoons_{r,t} \pi_B$.
\end{theorem}

\begin{proof} 
Given $m,\ell\in\N$, let  $r_0 = 0$, inductively define $r_{i+1} = 3r_i + 1$, and set $r = r_{m-1}$.
Further, let $t = m \cdot e + 1$, 
where $e \coloneqq 1 + \ell + \ell^2 + \dots +\ell^r$ is the maximal number of elements 
in an $r$-neighbourhood of a point, in $K$-interpretations with Gaifman graphs with maximal degree $\ell$.
Assume that  $\pi_A$ and $\pi_B$ are $K$-interpretations with that property,
such that $\pi_A \rightleftharpoons_{r,t} \pi_B$.

We construct an $m$-back-and-forth system $(I_j)_{j\leq m}$ for $(\pi_A,\pi_B)$ 
by setting
\[
    I_{j} \coloneqq \{ \ta\mapsto\tb: |\ta|=|\tb| = m-j \text{ and }B_{r_{j}}^{\pi_A}(\ta) \cong B_{r_{j}}^{\pi_B}(\tb)\}.
\]
We have $I_m=\{\emptyset\}$, and since $\pi_A \rightleftharpoons_{r,t} \pi_B$,
we have for every $a\in A$ some $b\in B$, and vice versa, such that
$B_{r}^{\pi_A}(a) \cong B_{r}^{\pi_B}(b)$, so $I_m$ has back-and-forth extensions in
$I_{m-1}$.
Consider now a partial isomorphism $\ta\mapsto\tb$ in $I_{j+1}$. There is
an isomorphism $\rho\colon B_{3r_j+1}^{\pi_A}(\ta) \cong B_{3r_j+1}^{\pi_B}(\tb)$.
By symmetry, it suffices to prove the forth-property: for every $a\in A$ we must find
some $b\in B$ such that $\ta a\mapsto \tb b \in I_j$ which means that
$B_{r_{j}}^{\pi_A}(\ta a) \cong B_{r_{j}}^{\pi_B}(\tb b)$.

{\it Case 1 ($a$ close to $\ta$).}
If $a \in B^{\pi_a}_{2r_j+1}(\ta)$, then we choose 
$b = \rho(a) \in B^{\pi_B}_{2r_j+1}(\tb)$.
This is a valid choice since $B^{\pi_A}_{r_j}(\ta a)\subseteq B^{\pi_A}_{3r_j+1}(\ta)$
so $\rho$ also provides an isomorphism between $B_{r_{j}}^{\pi_A}(\ta a)$ and $B_{r_{j}}^{\pi_B}(\tb b)$.

{\it Case 2 ($a$ far from $\tup a$).}
If $a \not\in B^{\pi_a}_{2r_j+1}(\ta)$, then 
$B^{\pi_A}_{r_j}(a) \cap B^{\pi_A}_{r_j}(\ta) = \emptyset$.
Hence, it suffices to find $b \in B$ 
such that $B^{\pi_B}_{r_j}(b)$ has the same 
isomorphism type as $B^{\pi_A}_{r_j}(a)$  (call this $\iota$) with the property that
$b$ has distance at least $2r_j+2$ to $\tb$. Since
$\pi_A$ and $\pi_B$ only track positive information the isomorphisms can be
combined by \cref{lem-topi} to show that  $B_{r_{j}}^{\pi_A}(\ta a) \cong B_{r_{j}}^{\pi_B}(\tb b)$.

Assume that no such $b$ exists.
Let $s$ be the number of elements realising  $\iota$ in $\pi_B$.
Since all of them are have distance at most $2r_j+1$ from $\tb$
and there are at most $t$ elements in $r$-neighbourhoods
around $\tb$, we have that $s\leq t$.
On the other side there are at least $s+1$ elements realising $\iota$ 
in $\pi_A$, namely $s$ elements in $B^{\pi_A}_{2r_j+1}(\ta)$
(due to $\rho$) and $a$.
But this contradicts the fact that $\iota$ either has the same number of realisations
in $\pi_A$ and $\pi_B$, or more than $t$ realisations in both interpretations. 
Hence such an element $b$ exists, and we have proved that
$(I_j)_{j\leq m}$  is indeed a $m$-back-and-forth system for $(\pi_A,\pi_B)$.

By \cref{prop-back-and-forth} this implies that $\pi_A\equiv_m\pi_B$.
\end{proof}

On the other side, we observe that Hanf's locality theorem in general \emph{fails} for semirings
with non-idempotent operations due to the possibility to count.

\begin{Example}[Counterexample Hanf]
Consider the natural semiring $(\N,+,\cdot,0,1)$ and $\psi = \E x \, U x$ over signature $\tau = \{U\}$.
For each $n$, we define a model-defining $K$-interpretation $\pi_n$ with universe $\{a_1,\dots,a_n\}$ by setting $\pi(U a_i) = 1$ for all $i$.
Then $\pi_n \ext \psi = \sum_{i} \pi(U a_i) = n$.

As we only have unary predicates, all neighbourhoods are trivial.
That is, they consist of just one element and all of them have the same isomorphism type.
Thus, $\pi_n$ realises this single isomorphism type precisely $n$ times,
which means that $\pi_n \rightleftharpoons_{r,t} \pi_t$ for all $r,t$ with $n\geq t$.
But $\pi_n \ext \psi \neq \pi_t \ext \psi$ for $n \neq t$, so Hanf's theorem fails for the natural semiring.

This example readily generalises to all semirings containing an element $s \in K$ for which there are arbitrarily large numbers $n,m \in \N$ with $m \cdot s \neq n \cdot s$ or $s^m \neq s^n$ ($m \cdot s$ and $s^m$ refer to the $m$-fold addition and multiplication of $s$, respectively).
Indeed, we can map all atoms $U a_i$ to $s$ and observe that Hanf's theorem fails for either $\psi = \E x \, U x$ or $\psi = \A x \, U x$.
\end{Example}

\section{Gaifman Normal Forms in Semirings Semantics}
\label{secGaifmanDefinitions}

We briefly recall the classical notion of Gaifman normal forms (cf.\ \cite{Gaifman82,EbbinghausFlu99}), which capture locality in a syntactic way.
Gaifman normal forms are Boolean combinations of \emph{local formulae} $\phi \loc r(x)$ and \emph{basic local sentences}.
A local formula $\phi \loc r(x)$ is a formula in which all quantifiers are \emph{relativised} to the $r$-neighbourhood of $x$,
for instance $\E y \, \theta(x,y)$ is relativised to $\E y (d(x,y) \le r \land \theta(x,y))$.
Here, $d(x,y) \le r$ asserts that $x$ and $y$ have distance $\le r$ in the Gaifman graph, which can easily be expressed in first-order logic (in Boolean semantics).
A basic local sentence asserts that there exist \emph{scattered} elements, i.e., elements with distinct $r$-neighbourhoods, which all satisfy the same $r$-local formula: $\E x_1 \dots \E x_m (\bigland_{i \neq j} d(x_i,x_j) > 2r \land \bigland_i \phi \loc r(x_i))$.
By Gaifman's theorem, every formula has an equivalent Gaifman normal form, which intuitively means that it only makes statements about distinct local neighbourhoods.

Moving to semiring semantics, we keep the notion of Gaifman normal forms close to the original one, with two exceptions.
First, we only consider formulae in negation normal form.
This means that we restrict to \emph{positive} Boolean combinations and, in turn, permit the duals of basic local sentences (i.e., the negations of basic local sentences, in negation normal form).
Second and most importantly, we lose the ability to express relativised quantifiers\footnote{We could use the same formula for $d(x,y) \le r$ as in the Boolean case. However, this formula would not just evaluate to $0$ or $1$, but would include the values of all edges around $x$, so each relativised quantifier would have the unintended side-effect of multiplying with the edge values in the neighbourhood. One can show that this side-effect would make Gaifman normal forms impossible (see \cref{appendixDistance} for details).} in our logic.
Instead, we extend first-order logic by adding relativised quantifiers (\emph{ball quantifiers}) of the form $\Qballtau \tau y r x$ for $Q \in \{\E,\A\}$ with the following semantics:
given a formula $\phi(x,y)$, a $K$-interpretation $\pi \from \Lit_A(\tau) \to K$, and an element $a$, we define
\[
    \pi\ext{\Eballtau \tau y r a\; \phi(a,y)} \coloneqq \sum_{\mathclap{b\in \ballpi[\pi] r a}} \,\pi\ext{\phi(a,b)}, \qquad
    \pi\ext{\Aballtau \tau y r a\; \phi(a,y)} \coloneqq \prod_{\mathclap{b\in \ballpi[\pi] r a}} \,\pi\ext{\phi(a,b)}.
\]
Notice that, similar to distance formulae in Boolean semantics, we define ball quantifiers for a fixed signature $\tau$.
For an interpretation $\pi \from \Lit_A(\tau') \to K$ of a larger signature $\tau' \supseteq \tau$, we define the semantics of $\Qballtau \tau z r a$ by considering the $r$-neighbourhood only with respect to the $\tau$-reduct of $\pi$, so that only relations in $\tau$ are relevant.
We drop $\tau$ and write $\Eball y r a$ or $\Aball y r a$ if the signature is clear from the context.

\medskip
This alone is not as expressive as the Boolean notion.
Indeed, the ability to express $d(x,y) \le r$ in Boolean semantics leads to the following properties of $r$-local formulae:
\begin{itemise}
\item \emph{Increasing the radius:}
Every $r$-local formula $\phi \loc r(\tx)$ is equivalent to an $r'$-local formula $\phi \loc {r'}(\tx)$ for any $r' > r$,
as we can replace $\Eball y r x$ by $\Eball y {r'} \tx (d(x,y) \le r \land \dots)$.

\item \emph{Requantifying:}
The formula $\psi(x) = \Eball y {r'} x \, \phi \loc r (y)$, where $\phi \loc r$ may contain relativised quantification around $y$, is equivalent to an $(r+r')$-local formula $\hat\psi \loc {r+r'} (x)$, where only relativised quantifiers around $x$ are permitted (with radius $r+r'$).
Indeed, we can replace any quantifiers $\Eball z r y$ in $\phi \loc r (y)$ by $\Eball z {r+r'} x (d(y,z) \le r \land \dots)$.
\end{itemise}

In order to capture these properties with ball quantifiers, we consider the \emph{quantification dag}
$D(\phi)$ of a formula $\phi(\tx)$ which contains nodes for all variables in $\phi$ and where for every quantifier $\Qball z {r'} y$ in $\phi$, we add an edge $z \to y$ with distance label $r'$ (see \cref{figDag}).
We define $\phi$ to be $r$-local if the summed distance of any path ending in a free variable $x \in \tx$ is at most $r$.

\begin{figure}
\begin{center}
$\displaystyle \phi \loc r (x,y) {{}={}} \Aball{z_1}{r_1}{x} \, \big( \Eball{z_2}{r_2}{z_1} \; \Aball{z_3}{r_3}{z_1} \, \neg E z_2 z_3 \big) \,\lor\, \Aball{z_4}{r_4}{y} \, E x z_4$
\end{center}
\begin{center}
\begin{tikzpicture}[scale=1]
 \node (x) at (0,0) {$x$};
 \node (y) at (3,0) {$y$};
 \node (z1) at (0,-1) {$z_1$};
 \node (z2) at (-.5,-2) {$z_2$};
 \node (z3) at (+.5,-2) {$z_3$};
 \node (z4) at (3,-1) {$z_4$};
 \draw [blue!70!black,->,>=stealth',every node/.style={pos=0.4,inner sep=2pt,font=\scriptsize}]
  (z1) edge node[right] {$r_1$} (x)
  (z2) edge node[above left] {$r_2$} (z1)
  (z3) edge node[above right] {$r_3$} (z1)
  (z4) edge node[right] {$r_4$} (y)
  ;
 \clip (-3,.5) rectangle (6,-2.35);
 
 \fill[blue!70!black,opacity=0.1] (0,0) circle (2.3);
 \fill[blue!70!black,opacity=0.1] (3,0) circle (2.3);
 \draw [dashed,blue!70!black] (y) -- node [sloped,above,pos=0.8,inner sep=2pt,font=\scriptsize] {$r$} ++(-10:2.3cm);
\end{tikzpicture}
\end{center}
\caption{Example of a local formula and the corresponding quantification dag $D(\phi)$, with circles indicating $\ball r {xy}$. In this example, $\phi \loc r (x,y)$ is $r$-local for all $r \ge \max(r_1+r_2, r_1+r_3, r_4)$.}
\label{figDag}
\end{figure}
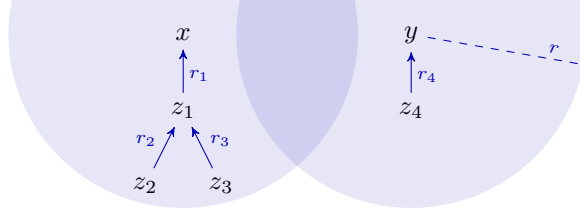

\begin{definition}[Local formula]
\label{defLocalFormula}
An \emph{$r$-local $\tau$-formula around $\tx$}, denoted $\phi \loc r (\tx)$, is built from $\tau$-literals
by means of $\land$, $\lor$ and ball quantifiers
$\Qballtau \tau z {r'} y$
such that in the associated quantification dag $D(\phi)$, all paths ending in a free variable $x \in \tx$ have total length at most $r$.

A bound variable of $\phi \loc r(\tx)$ is said to be \emph{locally quantified around $x \in \tx$}, if it is connected to $x$ in $D(\phi)$.
We sometimes write $\phi \loc r (\tx \mid \ty)$ to indicate that $\phi$ is an $r$-local formula around $\tx\ty$ where bound variables are locally quantified only around $\tx$, but not around $\ty$.
\end{definition}

We emphasise that in the Boolean case, \cref{defLocalFormula} is equivalent to the standard notion, so we do not add expressive power.
For convenience, we allow quantification $\Qball z {r'} \ty \, \phi(\ty,z)$ around a tuple $\ty$, which can easily be simulated by regular ball quantifiers.

For basic local sentences, we further need to quantify over scattered tuples.
To this end, we also add \emph{scattered quantifiers} $\Esc r \ty$ and $\Asc r \ty$ with the following semantics:
\[
    \pi \ext {\Esc r \ty \, \phi(\ty)} = \;\; \sum_{\mathclap{\substack{\ta \subseteq A\\d(a_i,a_j) > 2r \text{ for } i \neq j}}} \; \pi \ext {\phi(\ta)},
    \qquad
    \pi \ext {\Asc r \ty \, \phi(\ty)} = \;\; \prod_{\mathclap{\substack{\ta \subseteq A\\d(a_i,a_j) > 2r \text{ for } i \neq j}}} \; \pi \ext {\phi(\ta)}.  
\]

We remark that the addition of ball quantifiers makes it possible to express $d(x,y) \le r$ by a formula that only assumes values $0$ or $1$,
such as $\Eball {x'} {r} x \, (x' = y)$, which is $r$-local around $x$, or, if the semiring is idempotent, $\Eball {x'} {\frac r 2} x \, \Eball {y'} {\frac r 2} y \, (x'=y')$, which is $\frac r 2$-local around $xy$.
Analogously for $d(x,y) > r$, so we permit the use of distance formulae to simplify notation.
In absorptive semirings, which are the main focus of our positive results, we can then express scattered quantifiers as $\Esc r {y_1,\dots y_m} \, \theta(\ty) \coloneqq \E y_1 \dots \E y_m \big( \bigland_{i < j} d(y_i,y_j) > 2r \land \theta(\ty) \big)$ and $\Asc r {y_1,\dots y_m} \, \theta(\ty) \coloneqq \A y_1 \dots \A y_m \big( \biglor_{i < j} d(y_i,y_j) \le 2r \lor \theta(\ty) \big)$.

\begin{definition}[Local sentence]
\label{defLocalSentence}
A \emph{basic local sentence} is a sentence of the form
\[
    \Esc r {y_1,\dots y_m} \bigwedge_{i\leq m} \phi \loc r (y_i) \quad\text{ or }\quad
    \Asc r {y_1,\dots y_m} \bigvee_{i\leq m} \phi \loc r (y_i).
\]
An \emph{local sentence} is a positive Boolean combination of basic local sentences.
It is \emph{existential} if it only uses basic local sentences of the first kind (with existential scattered quantifiers).
\end{definition}

Based on these notions we can now formulate precisely the questions about Gaifman normal forms in 
semiring semantics:
\begin{bracketise}
\item For which semirings $K$ does every first-order \emph{sentence} have a $K$-equivalent local sentence?
\item For which semirings $K$ is it the case that every first-order \emph{formula} is $K$-equivalent to a positive Boolean combination of local formulae and basic local sentences?
\end{bracketise}

\section{Counterexamples Against Gaifman Normal Forms}
\label{secCounterexamples}

This section presents two examples for which a Gaifman normal form does not exist.
Both use the vocabulary $\tau = \{U\}$ with only unary predicates, so that the Gaifman graph $G(\pi)$ of any $K$-interpretation
$\pi\from\Lit_A(\tau)\ra K$ is trivial and the $r$-neighbourhood of a point, for any $r$, consists only of the point itself.
Thus, local formulae $\phi \loc r (x)$ around $x$ can always be written as positive Boolean combinations of literals $Ux$, $\neg Ux$
and equalities $x=x$, $x \neq x$.
Scattered tuples are simply distinct tuples, so basic local sentences take the form
\[
    \Edistinct {x_1,\dots,x_m} \bigwedge_{i\leq m} \phi \loc r (x_i)\quad\text{  or }\quad \Adistinct {x_1,\dots,x_m} \bigvee_{i\leq m} \phi \loc r (x_i).
\]

\subsection{A Formula Without a Gaifman Normal Form}
\label{secCounterexampleFormula}

Consider the formula $\psi(x) \coloneqq \E y(Uy\land y\neq x)$ which, in classical Boolean semantics, has the
Gaifman normal form $\phi(x) \coloneqq \Edistinct {y,z} \big( (Uy\land Uz) \lor (\neg Ux\land \E y Uy) \big)$.
However, in semiring semantics it is in general not the case that $\psi(x)\equiv_K\phi(x)$.
Indeed, for a semiring interpretation $\pi\from\Lit_A(\{U\})\ra K$ and $a\in A$ we have that
\[
    \pi\ext{\psi(a)}=\sum_{b\neq a}\pi(Ub)\qquad\text{but} \qquad
    \pi\ext{\phi(a)}=\sum_{\{b,c\}\subseteq A\atop b\neq c} \pi(Ub)\pi(Uc) + \pi(\neg Ua)\sum_{b}\pi(Ub).
\]
Here we consider the specific case of a universe with two elements $A=\{a,b\}$ and model-defining $K$-interpretations $\pi_{st}$
with $\pi_{st}(Ua)=s$ and $\pi_{st}(Ub)=t$, where $s,t\in K\setminus \{0\}$ and $s\neq t$.
Then $\pi_{st}\ext{\psi(a)}=t$ but $\pi_{st}\ext{\phi(a)}=st+ts$.
So, unless $K$ is the Boolean semiring, we find elements $s,t$ where
$\pi_{st}\ext{\psi(a)}\neq \pi_{st}\ext{\phi(a)}$, e.g.\ $t=1$, $s=2$ in the natural semiring.

Of course, it might still be the case that there is a different Gaifman normal form
of $\psi(x)$ for  semiring interpretations in a specific semiring $K$.
We prove that this is not the case.

\begin{proposition}
\label{thmCounterexampleFormula}
In any naturally ordered semiring with at least three element, the 
formula $\psi(x) = \E y(Uy\land y\neq x)$ does not have a Gaifman normal form.
\end{proposition}

For the proof, we describe the values that the building blocks of Gaifman normal forms may assume in $\pi_{st}$.
Recall that a local formula $\alpha(x)$ is equivalent to a positive Boolean combination of literals $Ux$, $\neg Ux$, and equalities.
Since $\pi_{st}(\neg Ux) = 0$ for all $x \in A$, we can ignore negative literals and thus view the evaluation $\pi_{st} \ext {\alpha(a)}$ as an expression built from the
semiring operations, the value $\pi_{st}(Ua) = s$ and constants $0,1$.
Similarly, $\pi_{st} \ext {\alpha(b)}$ is evaluated in the same way, but using $\pi_{st}(Ub) = t$ instead of $s$.
Hence there is a polynomial $p_\alpha(X) \in K[X]$ such that $\pi_{st}\ext{\alpha(a)} = p_\alpha(s)$ and $\pi_{st}\ext{\alpha(b)} = p_\alpha(t)$, for all choices of $s,t$.
For example, $\alpha \loc r (x) = Ux \lor \Eballtau \tau y r x (x \neq y \land Uy)$ can equivalently be written as $Ux \lor (x \neq x \land Ux)$,
which corresponds to the polynomial $p(X) = X + (0 \cdot X) = X$.
For the evaluation of a basic local sentence $\beta=\Edistinct{y,z} (\alpha(y)\land \alpha(z))$,
we have $\pi_{st} \ext \beta = \pi_{st} \ext {\alpha(a) \land \alpha(b)} + \pi_{st} \ext {\alpha(b) \land \alpha(a)} = p_\alpha(s) p_\alpha(t) + p_\alpha(t) p_\alpha(s)$.
That is, $\beta$ can be described by a polynomial $p_\beta(X,Y) \in K[X,Y]$ such that $\pi_{st} \ext \beta = p_\beta(s,t)$ and $p_\beta$ is \emph{symmetric} (that is, $p_\beta(X,Y) = p_\beta(Y,X)$).
The same holds for universal basic local sentences $\beta=\Adistinct{y,z} (\alpha(y)\lor \alpha(z))$.

It follows that we can represent every formula $\phi(x)$ in Gaifman normal form by a polynomial
$f_\phi(X,Y) = \sum_{i} h_i(X) g_i(X,Y)$, where the $g_i$ are symmetric,
such that $\pi_{st}\ext{\phi(a)} = f_\phi(s,t)$ for all $s,t$.
\Cref{thmCounterexampleFormula} then follows from the following algebraic observation.

\begin{lemma}  Let $K$ be a naturally ordered semiring with at least three elements.
For any polynomial $f(X,Y) = \sum_{i} h_i(X) g_i(X,Y)$ 
where the $g_i$ are symmetric polynomials, there
exist values $s,t\in K \setminus \{0\}$ such that $f(s,t)\neq t$.
\end{lemma}

\begin{proof} We first consider the case that the semiring $K$ is \emph{absorptive}, that is, $1 + r = 1$ for all $r \in K$ (or, equivalently, that $1$ is the maximal element).
Pick any $\eps \in K \setminus \{0,1\}$.
Then $0 < \eps < 1$ by minimality/maximality of $0,1$.
Consider $(s,t) = (\eps,1)$.
If $f(s,t) \neq 1$ we are done, so suppose $f(s,t) = \sum_i h_i(\eps) \cdot g_i(\eps,1) = 1$.
We now switch to $(s,t) = (1,\eps)$.
By symmetry, we have $g_i(\eps,1) = g_i(1,\eps)$ for all $i$.
Note that both semiring operations, and hence all polynomials, are monotone w.r.t.\ the natural order.
We can thus conclude
\[
    f(1,\eps) = \sum_i h_i(1) \cdot g_i(1,\eps) \ge \sum_i h_i(\eps) \cdot g_i(1,\eps) = \sum_i h_i(\eps) \cdot g_i(\eps,1) = 1 > \eps.
\]

Now assume that $K$ is naturally ordered, but \emph{not} absorptive. 
We claim that there is $t\in K$ with $t > 1$.
Indeed, since $K$ is not absorptive there is $r \in K$ with $1 + r \neq 1$.
Since $1 + r \ge 1$ by definition of the natural order, we must have $t \coloneqq 1 + r > 1$.
If $f(1,t) \neq t$ we are done, so suppose $f(1,t) = t$.
By symmetry of $g_i$ and monotonicity, we conclude
\[
    f(t,1) = \sum_i h_i(t) \cdot g_i(t,1) \ge \sum_i h_i(1) \cdot g_i(1,t) = f(1,t) = t \;>\; 1. \qedhere
\]
\end{proof}

\subsection{A Sentence Without a Gaifman Normal Form}
\label{secCounterexampleTropical}

While Gaifman normal forms need not exist for formulae, in all relevant semirings beyond the Boolean one,
they might still exist for sentences.
Indeed, we shall prove a positive result for min-max semirings.
However, such a result seems only possible for semirings where both operations are idempotent, similar to Hanf's theorem.
For other semirings one can find rather simple counterexamples, as we illustrate for the tropical semiring $\Trop= (\Real_{+}^{\infty}, \min, +, \infty, 0)$.
Notice the unusual operations and neutral elements: in particular, true equalities are interpreted by $0$ and false equalities by $\infty$.

\begin{proposition} The sentence $\psi \coloneqq \exists z \forall x \exists y (Uy \vee x = z)$
has no Gaifman normal form in the tropical semiring.
\end{proposition}

\begin{proof} We consider,  for every $n\geq 1$, the
$\Trop$-interpretations $\pi_n$ with universe $\{a_1,\dots,a_n\}$ such that
$\pi_n(Ua_i)=1$  (and hence $\pi_n(\neg Ua_i)=\infty$) for all $n$ and all $i$.
Clearly $\pi_n\ext\psi= n-1$ for all $n\geq 1$.

Recall that every local formula $\phi^{(r)}(x)$ is just a positive Boolean combination of literals $Ux$, $\neg Ux$,
and $x=x$, which means that there exists a fixed constant $c_\phi\in\N\cup\{\infty\}$ such that 
$\pi_n\ext{\phi^{(r)}(a_i)}=c_\phi$ for all $n$ and all $i$. It follows that for the basic local sentences
we have that 
\begin{align*}
    \pi_n\ext{ \Edistinct {x_1,\dots,x_m} \bigwedge_{i\leq m} \phi \loc r (x_i)}&=m c_\phi,\\ 
    \pi_n\ext{\Adistinct {x_1,\dots,x_m} \bigvee_{i\leq m} \phi \loc r (x_i)}&= n \cdot (n-1) \cdots (n-m+1) \cdot c_\phi \ge n c_\phi.
\end{align*}
For a local sentence $\eta$, we thus have that $\pi_n \ext \eta$ is an expression built from the operations $\min$ and $+$, terms of the form $n \cdots (n-m+1)$ and constants $c \in \N \cup \{\infty\}$.
Hence either $\pi_n \ext \eta \ge n$ for all sufficiently large $n$, or $\pi_n \ext \eta = c$ for some constant $c$.
Thus $\eta\not\equiv_{\Trop}\psi$. 
\end{proof}

A similar construction works for the natural semiring $(\N,+,\cdot,0,1)$
and we conjecture that it can be adapted to any infinite semiring with operations that are not idempotent.

\section{Gaifman's Theorem for Min-Max Semirings}
\label{secGaifmanProof}

In this section, we prove our main result: a version of Gaifman's theorem for sentences evaluated in min-max semirings (which can be lifted to lattice semirings, see \cref{secStrengthening}).
We write $\Minmax$ for the class of min-max semirings and refer to $\meq$ as \emph{minmax-equivalence}.
Further, we use the notation $\mle$ (and similarly $\mge$), where $\phi \mle \psi$ means that $\pi \ext {\phi} \le \pi \ext {\psi}$ for every $K$-interpretation $\pi$ in a min-max semiring $K$.

\begin{theorem}[Gaifman normal form]
\label{thmGaifman}
Let $\tau$ be a finite relational signature.
Every $\FO(\tau)$-sentence $\psi$ is minmax-equivalent to a local sentence.
\end{theorem}

Contrary to Hanf's locality theorem, we cannot follow the classical proofs of Gaifman's theorem.
For instance, the proof in \cite{EbbinghausFlu99} is based on the Ehrenfeucht-Fra\"{\i}ss\'e method and makes use of characteristic sentences, which in general do not exist in semiring semantics over min-max semirings (cf.\ \cite{GraedelMrk21}).
Gaifman's original proof \cite{Gaifman82} is a constructive quantifier elimination argument (which is similar to our approach), but makes use of negation to encode case distinctions in the formula.
In semiring semantics, this is only possible to a limited degree (e.g., $\theta \lor \neg\theta$ is not guaranteed to evaluate to $1$, so we may have $\phi \not\meq \phi \land (\theta \lor \neg\theta)$).
Another argument why Gaifman's proof does not go through is that it applies to formulae, whereas formulae need not have Gaifman normal forms in our setting (cf.\ \cref{secCounterexampleFormula}).

Instead, we present a novel proof of Gaifman's theorem that applies to the Boolean case as well as to min-max semirings.
While our strategy is similar to Gaifman's -- a constructive elimination of quantifier alternations -- we have to phrase all results in terms of sentences and need to be more careful to derive equivalences that hold in all min-max semirings.
This turns out to be surprisingly difficult, but in the end also leads to a stronger version of Gaifman's classical result (see \cref{secStrengthening}).

\subsection{Toolbox}

The proof consists of a series of rather technical equivalence transformations, but is based on a few key observations that we present below.
For the remaining section, we fix a finite relational signature $\tau$ and a min-max semiring $K$, unless stated otherwise.

\subparagraph*{Normal forms.}
All min-max semirings share algebraic properties such as distributivity and idempotency of both operations with the Boolean semiring, so many classical equivalences of first-order logic are also minmax-equivalences.
For instance, every semiring is distributive which implies that $\land$ distributes over $\lor$.
In min-max semirings, we further have the dual property ($\max$ distributes over $\min$), thus also $\phi_1 \lor (\phi_2 \land \phi_3) \meq (\phi_1 \lor \phi_2) \land (\phi_1 \lor \phi_3)$.
This means that we can transform any positive Boolean combination into a minmax-equivalent \emph{disjunctive normal form} (DNF), and into \emph{conjunctive normal form} (CNF).
Moreover, one easily verifies that we can push quantifiers to the front (tacitly assuming distinct variable names), so it suffices to prove \cref{thmGaifman} for sentences in \emph{prenex normal form}.
We list the relevant minmax-equivalences below.

\begin{lemma}
\label{thmAxioms}
For all FO-formulae,
\begin{enumerate}
\item $\phi \land \phi \meq \phi \lor \phi \meq \phi$, \label{axiomIdempotent}
\item $\phi \lor (\phi \land \psi) \meq \phi \land (\phi \lor \psi) \meq \phi$, \label{axiomAbsorbtion}
\item $\phi_1 \land (\phi_2 \lor \phi_3) \meq (\phi_1 \land \phi_2) \lor (\phi_1 \land \phi_3)$, \label{axiomDistrib}
\item $\phi_1 \lor (\phi_2 \land \phi_3) \meq (\phi_1 \lor \phi_2) \land (\phi_1 \lor \phi_3)$, \label{axiomDualDistrib}
\item $\phi \circ (\E x \, \theta) \meq \E x (\phi \circ \theta)$, for $\circ \in \{\lor,\land\}$, if $x$ does not occur free in $\phi$, \label{axiomExistsDistrib}
\item $\phi \circ (\A x \, \theta) \meq \A x (\phi \circ \theta)$, for $\circ \in \{\lor,\land\}$, if $x$ does not occur free in $\phi$. \label{axiomForallDistrib}
\item $\E x \, (\phi \lor \theta) \meq (\E x \,\phi) \lor (\E x \,\theta)$, and analogously for $\A$ and $\land$.
\end{enumerate}
\end{lemma}
\begin{proof}
Direct consequences of the algebraic properties of min-max semirings: idempotence, absorption, distributivity and commutativity.
\end{proof}

\subparagraph*{Symmetry.}
It will be convenient to exploit the inherent symmetry of min-max semirings to simplify the following proofs.
For example, the proof of \cref{thmGaifman} for sentences of the form $\E \ty \A \tx \ \phi \loc r (\ty,\tx)$ is completely symmetric to the one for sentences $\A \ty \E \tx \ \phi \loc r (\ty,\tx)$, by exchanging minima and maxima in the proof.

This insight can be formalised. With a formula $\phi$ we associate its \emph{dual} $\dual \phi$, resulting from $\phi$ by swapping $\E$, $\A$ (including scattered and ball quantifiers) and also $\lor$, $\land$, e.g.:
\[
    \dual{(\Esc r {x,y} \; \Aball z r x \, (\theta(z) \lor \psi(y))} = \Asc r {x,y} \; \Eball z r x \, (\dual\theta(z) \land \dual\psi(y)).
\]

\begin{lemma}[Symmetry]
\label{thmSymmetry}
If a sentence $\psi$ is minmax-equivalent to a local sentence, then so is $\dual\psi$.
\end{lemma}

\begin{proof}
First observe that if $\phi(\tx)$ is $r$-local around $\tx$, then so is $\dual\phi(\tx)$.
Similarly, if $\phi$ is a basic local sentence, then so is $\dual\phi$.
Fix any min-max semiring $K = (K,\max,\min,0,1)$.
We define its \emph{dual} as $\dual K = (K,\min,\max,1,0)$, i.e.\ by inverting the underlying order.
Since $K$ and $\dual K$ are both min-max semirings and share the same domain, we can interpret any $K$-interpretation $\pi$ also as $\dual K$-interpretation, which we denote as $\dual \pi$.

Now let $\psi$ be minmax-equivalent to a local sentence.
The equivalence holds in particular for $\dual K$, and by putting the local sentence in disjunctive normal form, we have
\[
    \psi \keq[\dual K] \biglor_i \bigland_j \phi_{ij}, \qquad
    \text{with basic local sentences } \phi_{ij}. \tag{$\dagger$}
\]
We claim that
\[
    \dual\psi \keq[K] \bigland_i \biglor_j \dual\phi_{ij},
\]
which prove the lemma.
To prove the claim, let $\pi$ be a $K$-interpretation (over universe $A$).
We first note that for any FO-formula $\phi(\tx)$ and any tuple $\ta \subseteq A$ of matching arity, $\pi \ext {\phi(\ta)} = \dual\pi \ext {\dual\phi(\ta)}$ by a straightforward induction on $\phi$.
Then
\[
    \pi \ext {\dual\psi} =
    \dual\pi \ext {\dual{(\dual\psi)}} =
    \dual\pi \ext {\psi} \overset{(\dagger)}{=}
    \dual\pi \Ext {\biglor_i \bigland_j \phi_{ij}} =
    \pi \Ext {\bigland_i \biglor_j \dual\phi_{ij}}. \qedhere
\]
\end{proof}

\subparagraph*{Logic.}
Contrary to Gaifman's original proof, our results only apply to sentences (cf.\ \cref{secCounterexampleFormula}).
However, we still need to apply intermediary results to subformulae involving free variables.
This is possible by the following lemma, which allows us to temporarily replace atoms involving the undesired variables with fresh relation symbols.

Intuitively, we think of the new relation as an abstraction of the original atoms.
All equivalence transformations that we perform in the abstract setting also hold, in particular, in the original setting.
However, there is one caveat: in order to apply our results in the abstract setting, we must also extend the ball quantifiers to the new signature.
Since we do not want to change the semantics of the quantifiers, we thus restrict the lemma to situations where the new relation does not add any edges to the Gaifman graph.

\newcommand{\sig}[1]{^{\langle{#1}\rangle}}

We use the following notation.
Let $\phi$ be a $\tau$-formula, including ball quantifiers of the form $\Qballtau \tau y r x$.
Given a different signature $\tau'$, we write $\phi \sig {\tau'}$ for the formula where all $\Qballtau \tau y r x$ are replaced with $\Qballtau {\tau'} y r x$ (assuming this is well-defined).
Further, we write $\phi[\theta(\tx,\ty) / R\tx]$ to replace every occurrence of $\theta(\tz,\ty)$ for any variable tuple $\tz$ by the atom $R\tz$.

\begin{lemma}[Abstraction]
\label{thmAbstraction}
Let $K$ be an arbitrary semiring, $\phi(\ty)$ a $\tau$-formula and $\theta(\tx,\ty)$ a subformula of $\phi(\ty)$.
Let $R$ be a fresh relation symbol with arity matching $\tx$, and set $\tau' = \tau \cup \{R\}$.
If we have
\begin{itemise}
\item $\phi [\theta(\tx,\ty) / R\tx]\sig {\tau'}(\ty) \allowbreak\keq \psi(\ty)$ for a $\tau'$-formula $\psi$ (abstract setting $\tau'$), and
\item $R$ is unary or $\theta = P \tz$ is a positive literal with $\tx \subseteq \tz$,
\end{itemise}
then also $\phi(\ty) \keq \psi [R\tx / \theta(\tx,\ty)] \sig \tau (\ty)$ (original setting $\tau$).
This remains true if we only replace some occurrences of $\theta$ by $R$, but not all.
\end{lemma}

\begin{proof}
Intuitively, if the equivalence holds for a fresh relation symbol $R$, then it holds for any interpretation of $R$, and hence in particular for the interpretation defined by $\theta$.

Formally, assume that $\phi \sig {\tau'} [\theta(\tx,\ty) / R(\tx)](\ty) \keq \psi(\ty)$.
Let $\pi$ be a $K$-interpretation over universe $A$ and signature $\tau$, and let $\tb \subseteq A$.
We have to show that $\pi \ext {\phi(\tb)} = \pi \ext {\psi \sig \tau [R(\tx) / \theta(\tx,\ty)](\tb)}$.
Extend $\pi$ to a $K$-interpretation $\pi'_{\tb}$ over $A$ and $\tau'$ by setting
\[
    \pi'_{\tb}(R\ta) \coloneqq \pi\ext{\theta(\ta,\tb)}, \quad \text{for all $\ta \subseteq A$}.
\]

We first show that the extension does not affect the Gaifman graph, i.e., $G(\pi) = G(\pi'_{\tb})$.
This is clearly the case if $R$ is unary.
Otherwise, we have $\theta(\tx,\ty) = P\tz$ and $\tx \subseteq \tz$ by assumption.
The extension can only add new edges to the Gaifman graph, so assume that $G(\pi'_{\tb})$ contains an additional edge between two elements $c,d$.
Hence there is a tuple $\ta$ with $c,d \in \ta$ and $\pi'_{\tb}(R\ta) \neq 0$.
Since $\tx \subseteq \tz$, the elements $c,d$ also occur in the atom $\theta(\ta,\tb)$ with $\pi \ext {\theta(\ta,\tb)} \neq 0$, hence the edge between $c,d$ is already present in $G(\pi)$.

Since the Gaifman graphs coincide, ball quantifiers $\Qballtau \tau y r x$ in $\pi$ and $\Qballtau {\tau'} y r x$ in $\pi'_{\tb}$ range over the same elements.
By a straightforward induction, we thus have
\begin{align*}
    \pi \ext {\phi(\tb)} &= \pi'_{\tb} \ext {\phi [\theta(\tx,\ty) / R\tx] \sig {\tau'} (\tb)}
        &&\text{by construction of $\pi'_{\tb}$ (induction on $\phi$)}, \\
    &= \pi'_{\tb} \ext {\psi(\tb)}
        &&\text{by the $K$-equivalence (abstract setting)}, \\
    &= \pi \ext {\psi[R\tx / \theta(\tx,\ty)] \sig\tau (\tb)}
        &&\text{by construction of $\pi'_{\tb}$ (induction on $\psi$)}. \qedhere
\end{align*}
\end{proof}

\subparagraph*{Locality.}

Concerning locality, we make two simple but crucial observations.
For the first one, consider a local formula $\phi \loc r (x,y)$ around two variables $x$ and $y$.
Such a formula may assert that $x$ and $y$ are close to each other, for instance $\phi \loc r (x,y) = Exy$.
But if $x$ and $y$ do not occur together within one literal, then $\phi \loc r$ intuitively makes independent statements about the neighbourhood of $x$, and the neighbourhood of $y$, so we can split $\phi \loc r$ into two separate local formulae.
For the general case $\phi \loc r (\tx)$ in several variables, we group $\tx$ into tuples $\tx^1,\dots,\tx^n$ with the idea that $\phi \loc r$ makes independent statements about each group $\tx^i$.

\begin{lemma}[Separation]
\label{thmSeparation}
Let $\phi \loc r (\tup x^1, \dots, \tup x^n)$ be a local formula around $\tup x^1 \dots \tup x^n$ and
define $X_i$ as the set of variables locally quantified around $\tx^i$ in $\phi \loc r$.
If each literal of $\phi \loc r (\tx^1, \dots, \tx^n)$ uses only variables in $\overline x^i \cup X_i$ for some $i$,
then $\phi \loc r (\tx^1,\dots,\tx^n)$ is minmax-equivalent to a positive Boolean combination of $r$-local formulae around each of the $\tx^i$. 
\end{lemma}

\begin{proof}
We write $\tx = \tx^1 \tx^2 \dots \tx^n$ and assume w.l.o.g.\ that $\phi$ is in prenex normal form,
so $\phi \loc r(\tx) = Q_1y_1 \dots Q_my_m \ \psi(\tx, \tup y)$ where $\psi$ contains no quantifiers
and each $Q_ky_k$ is of the form $\Eball{y_k}{r_k}{z_k}$ or $\Aball{y_k}{r_k}{z_k}$.
Let $X_i^k \coloneqq X_i \cap \{y_1,\dots,y_k\}$ denote the subset of variables locally quantified around $\tx_i$ that are among the first $k$ bound variables.
We show by induction on $k$, from $m$ down to $0$, that there exists $\psi_k(\tx,y_1,\dots,y_k)$, a positive Boolean combination of $r$-local formulae around $\tx^1 \cup X_1^k$, $r$-local formulae around $\tx^2 \cup X_2^k$, \dots, such that 
$\phi \loc r(\tx) \meq Q_1y_1 \dots Q_ky_k \ \psi_k(\tx,y_1,\dots,y_k)$.
Then $k=0$ implies the lemma.

For $k=m$ this certainly holds: $\psi(\tx, \tup y)$ is quantifier-free and by assumption all literals are local formulae around $\tx^i \cup X_i$ for some $i$.
Now assume the claim already holds for $k > 0$.
We show that there is $\psi_{k-1}(\tx,y_1, \dots ,y_{k-1})$ of the required form such that
\[
    Q_1y_1 \dots Q_ky_k \ \psi_k(\tx,y_1, \dots ,y_k) \meq Q_1y_1 \dots Q_{k-1}y_{k-1} \ \psi_{k-1}(\tx,y_1, \dots ,y_{k-1}). \tag{$*$}
\]

It suffices to consider the case $Q_k y_k = \Eball{y_k}{r_k}{z_k}$ (the other case is symmetric).
Since $y_k$ is locally quantified around $z_k$, there is an index $i_0$ such that
either $y_k \in X_{i_0}$ and $z_k \in \tx^{i_0}$, or (if $z_k$ is locally quantified as well) $y_k,z_k \in X_{i_0}$.
We put $\psi_k$ in disjunctive normal form,
\[
    \psi_k(\tx,y_1, \dots ,y_k) \meq \biglor_j \bigland_{i=1}^n \phi_{j,i} \loc r(\tx^i, X_i^k),
\]
where each $\phi_{j,i}$ is a conjunction of literals (and hence a local formula) in $\tx^i \cup X_i^k$.
Then
\begin{align*}
    &\Eball {y_k} {r_k} {z_k} \ \psi_k(\tx,y_1, \dots ,y_k) \\[.5em] \meqtag{\ref{thmAxioms}}
    &\biglor_j \Big( \Eball {y_k} {r_k} {z_k} \ \bigland_i \phi_{j,i} \loc r(\tx^i, X_i^k) \Big) \\ \meqtag{\dagger}
    &\biglor_j \Big( \underbrace{\Eball {y_k} {r_k} {z_k} \big( \phi_{j,{i_0}} \loc r(\tx^{i_0}, X_{i_0}^k) \big)}_{\text{local formula around $\tx^{i_0} \cup X_{i_0}^{k-1}$}} \;\land\; \bigland_{i \neq i_0} \phi_{j,i} \loc r(\tx^i, X_i^{k-1}) \Big) \eqqcolon \psi_{k-1}(\tx,y_1, \dots ,y_{k-1}).
\end{align*}
For $(\dagger)$, recall that $y_k,z_k \in \tx^{i_0} \cup X_{i_0}^k$ and hence only $\phi_{j,{i_0}}$ may depend on $y_k$;
we can thus move the formulae $\phi_{j,i}$ for $i \neq i_0$ out of the scope of the quantifier (cf.\ \cref{thmAxioms}).
Finally, we remark that all local formulae in $\psi_{k-1}$ are of radius $r$, simply because their quantification dags are all contained in the quantification dag of the original formula $\phi \loc r$ of radius $r$.
We conclude that $\psi_k$ has the required form and satisfies $(*)$.
\end{proof}

The second observation is that we can perform a clustering of any tuple $(a_1,\dots,a_n) \in A^n$
into classes $I_1,\dots,I_k$ so that elements within one class have ``small'' distance to each other, whereas different classes are ``far apart''.
This simple combinatorial observation is a fruitful tool to construct Gaifman normal forms:
it becomes easy to quantify elements with a known clustering, and by the following lemma we can then do a disjunction over all clusterings.

\begin{definition}[Configuration]
\label{defConfiguration}
Let $\pi$ be a $K$-interpretation with universe $A$.
Let $P=\{I_1,\dots,I_k\}$ be a partition of $\{1,\dots,n\}$ and define representatives $i_l = \min I_l$ of each class.
We say that a tuple $(a_1,\dots,a_n) \in A^n$ is in \emph{configuration} $(P,r)$, if 
\begin{alphatise}
    \item $d(a_{i_l},a_{i}) \le 5^{n-k}r-r$, for all $i \in I_l$, $l \in \{1,\dots,k\}$,
    \item $d(a_{i_l},a_{i_{l'}}) > 4 \cdot 5^{n-k}r$, for all $l \neq l'$ (representatives are $(2 \cdot 5^{n-k}r)$-scattered).
\end{alphatise}
\end{definition}

Notice that the definition of ``small`` and ``far apart`` depends not only on $r$, but also on $n$ and the number of classes $k$.
It is easy to see that such a partition always exists: condition \enumref{a} remains true if we merge two classes violating \enumref{b}, so starting from $P=\{\{1\}, \dots, \{n\}\}$ we can prove the following lemma by merging classes until both conditions hold.

\begin{lemma}[Clustering]
\label{thmClustering}
Let $\pi$ be a $K$-interpretation on $A$.
For all tuples $(a_1,\dots,a_n) \in A^n$ and all $r \ge 1$, there is a partition $P$ such that $(a_1,\dots,a_n)$ is in configuration $(P,r)$.
\end{lemma}

In the following proofs, we make frequent use of the Clustering Lemma and always stick to the notation $P = \{I_1,\dots,I_k\}$
with representatives $i_l = \min I_l$ for each $1 \le i \le k$.
The set of all partitions of $\{1,\dots,n\}$ is denoted as $\Part(n)$.

\subsection{Proof of Gaifman's Theorem}

The heart of our proof is the elimination of quantifier alternations.
We split the proof into steps.
Each step proves, building on the previous ones,
that sentences of a certain fragment can be translated to minmax-equivalent
local sentences. These fragments consist of
\begin{bracketise}
\item sentences of the form $\Esc r {x_1,\dots x_m} \bigwedge_{i\leq m} \phi_i \loc r (x_i)$;
\item existential sentences  $\E \tx \, \phi \loc r(\tx)$; 
\item existential-universal sentences $\E \ty \A \tx \, \phi \loc r(\ty,\tx)$;
\item all first-order sentences (\cref{thmGaifman}).
\end{bracketise}

We first prove that \enumref{4} follows from \enumref{2} and \enumref{3}.
This is based on the following lemma, which will be used to swap quantifiers.

\begin{lemma}
\label{thmGaifmanQuantifierSwap}
A positive Boolean combination of basic local sentences is minmax-equivalent both to a sentence of the form $\E \ty \A \tx \, \phi \loc r(\tx,\ty)$, and of the form $\A \ty \E \tx \, \psi \loc r(\tx,\ty)$.
\end{lemma}

\begin{proof}
Recall that scattered quantifiers are just abbreviations for regular quantifiers and distance formulae.
Each basic local sentence uses only one type of quantifier (followed by a local formula).
When we bring the local sentence into prenex normal form, we can choose to first push all $\E$-quantifiers to the front, and then the $\A$-quantifiers (or vice versa).
The local formulae can all be combined into a single local formula (around all variables).
\end{proof}

\begin{proof}[Proof of \cref{thmGaifman}]
We prove by induction that every sentence of the form
\[
    \Psi = Q_1 \tx_1 \, \dots Q_{n+1} \tx_{n+1} \; \phi \loc r (\tx_n,\tx_{n+1} \mid \tx_1, \dots, \tx_{n-1})
\]
with alternating quantifiers $Q_i \in \{\E,\A\}$ (i.e., with $n$ quantifier alternations), is equivalent to a local sentence.
This implies \cref{thmGaifman} (consider a quantifier-free formula $\phi$).

The cases $n=0$ and $n=1$ are covered by \enumref{2} and \enumref{3} (possibly combined with the \SymmetryLemma{}).
For $n \ge 2$, consider the inner formula
\[
    \Phi(\tx_1,\dots,\tx_{n-1}) = Q_n \tx_n Q_{n+1} \tx_{n+1} \; \phi \loc r (\tx_n,\tx_{n+1} \mid \tx_1,\dots,\tx_{n-1}).
\]
We apply the \AbstractionLemma{} to successively replace all atoms (including equalities) involving the free variables $\tx_1,\dots,\tx_{n-1}$ with fresh relation symbols, resulting in a \emph{sentence} $\tilde\Phi$.
Assume $Q_n = \E$ and $Q_{n+1} = \A$ (the other case is symmetric).
We can now swap the quantifiers in $\tilde\Phi$ by applying \enumref{3} and the previous lemma.
After reversing the substitution of atoms, we get that $\Phi(\tx_1,\dots,\tx_n)$ is minmax-equivalent to a formula of the form $\A\ty \E\tz \, \theta\loc r(\ty,\tz \mid \tx_1,\dots,\tx_n)$.
Since $Q_{n-1} = \A$, we can write $\Psi$ as
\[
    \Psi \meq Q_1 \tx_1 \, \dots Q_{n-1} \tx_{n-1} \ty \; Q_{n} \tz \; \theta \loc r (\ty,\tz \mid \tx_1, \dots, \tx_{n-1}),
\]
which has only $n-1$ quantifier alternations. The claim follows by induction.
\end{proof}

\subsection{Asymmetric Basic Local Sentences}

To simplify the proofs of steps \enumref{2} and \enumref{3}, we first establish the following lemma.
This can be seen as generalising the notion of basic local sentences by permitting different local formulae for each scattered variable $x_i$.
Such sentences have been called \emph{asymmetric} basic local sentences in \cite{GroheWoe04,DawarGroKreSch06}.

\begin{lemma}[Step \enumref{1}]
\label{thmGaifmanGeneralizedBasic}
Every sentence of the form $\Esc r {x_1,\dots x_m} \bigwedge_{i\leq m} \phi_i \loc r (x_i)$ is minmax-equivalent to an existential local sentence.
\end{lemma}

\begin{proof}[Proof sketch ($m=2$)]
The proof is somewhat technical, so we first sketch the idea for the case $m=2$, i.e., for a sentence $\Phi = \Esc r {x,y} \big( \phi \loc r (x) \land \psi \loc r (y) \big)$.
Imagine that we want to approximate $\Phi$ with a single existential local sentence.
A first candidate would be $\Psi \coloneqq \Esc r {x,y} \big( (\phi \loc r(x) \lor \psi \loc r(x)) \land (\phi \loc r(y) \lor \psi \loc r(y)) \big)$,
where we take the disjunction of $\phi \loc r$ and $\psi \loc r$ in order to have the same local formula for $x$ and $y$.

Fix a $K$-interpretation $\pi$ and recall that $\pi \ext \Phi = \max_{a,b} \min(\pi \ext {\phi \loc r(a)}, \pi \ext {\psi \loc r(b)})$.
Since $\Psi$ replaces each local formula with the maximum over both local formulae, we have $\Phi \mle \Psi$.
But we do not have equivalence since the maximal value $\pi \ext {\Psi}$ could be assumed by choosing elements $a,b$ and then choosing the same formula, say $\phi$ for both $a$ and $b$.

To improve our approximation while preserving $\Phi \mle \Psi$, we set $\Psi' \coloneqq \Psi \land \E x \, \phi \loc r(x) \land \E y \, \psi \loc r(y)$.
Now suppose that the maximal value $\pi \ext \Psi$ is assumed by choosing $a,b$ with $d(a,b) > 4r$ (so $a,b$ are $2r$-scattered).
If the maximum is achieved by further choosing different local formulae for $a$ and $b$, then $\pi \ext \Phi = \pi \ext {\Psi'}$ as desired.
So suppose we choose the same formula, say $\phi$, for both $a$ and $b$.
Let $c$ be an element for which the maximal value $\pi \ext {\E y \ \psi \loc r(y)}$ of $\psi$ is assumed.
As $c$ cannot be in both $\ball {2r} a$ and $\ball {2r} b$, we know that $c$ forms a $2r$-scattered tuple with $a$ or $b$, say with $a$.
Thus $\pi \ext \Phi \ge \min(\pi \ext {\phi \loc r(a)}, \pi \ext {\psi \loc r(c)}) \ge \pi \ext {\Psi'}$, so again $\pi\ext{\Phi} = \pi\ext{\Psi'}$.
In other words, our overapproximation $\Psi'$ is exact on tuples $(a,b)$ where $a$ and $b$ are sufficiently far from each other.

To achieve equivalence, we can thus increase the radius in $\Psi'$ (from $r$- to $2r$-scattered) and cover the remaining cases where $a,b$ are close by an additional local sentence:
\begin{align*}
    \Phi \meq &\Big(\E x \, \phi \loc r(x) \;\land\; \E y \, \psi \loc r(y) \;\land\; \Esc {2r} {x,y} \, (\phi \loc r(x) \lor \psi \loc r(x)) \land (\phi \loc r(y) \lor \psi \loc r(y)) \Big) \\[.25em]
     &{}\lor \E x \, \Eball {y} {4r} {x} \, \big(\phi \loc r(x) \,\land\, \psi \loc r(y) \,\land\, d(x,y) > 2r \big).
\end{align*}
Here, the first line consists of 3 basic local sentences and the second line adds another basic local sentence to account for the ``close'' case.
\end{proof}

\begin{proof}[Proof of \cref{thmGaifmanGeneralizedBasic}]
We generalise the above idea to $m$ variables by applying the \ClusteringLemma{}.
Let $\Phi = \Esc r {x_1,\dots,x_m} \, \bigland_{i=1}^m \phi_i \loc r(x_i)$.
To simplify notation, we write $\tx = (x_1,\dots,x_m)$ and $\tx_{- i} = (x_1,\dots,x_{i-1},x_{i+1},\dots,x_m)$.
The proof is by induction on $m$.
For $m=1$ there is nothing to show.
For $m \ge 2$, we claim that $\Phi$ is minmax-equivalent to
\begin{align*}
    \Psi &\coloneqq{} \Big(
        \Esc {4r} {\tx} \bigland_{i=1}^m \biglor_{j=1}^m \phi_j \loc r(x_i)
        \;\;\land\;\;
        \bigland_{j=1}^m \Esc {4r} {\tx_{-j}} \, \bigland_{i \neq j} \phi_i \loc r (x_i)
    \Big) \;\lor \tag{1} \\
    &\qquad\qquad\quad \biglor_{\mathclap{\substack{%
        P = \{I_1,\dots,I_k\} \,\in\, \Part(m),\\
        P \neq \{ \{1\},\dots,\{m\}\}
    }}} \quad
    \Esc {(2 \cdot 5^{m-k} 4r)} {x_1,\dots,x_k} \bigland_{l=1}^k \phi_{P,l}\loc {5^{m-k}4r}(x_l), \tag{2}
\end{align*}
where % may use \tertext{where} to align both equations
\begin{align*}    
    \phi_{P,l} \loc {5^{m-k}4r} (x_{i_l}) &\coloneqq \Eball {(x_i)_{i \in I_l\setminus \{i_l\}}} {5^{m-k}4r-r} {x_{i_l}} \ \Big(\bigland_{i \in I_l} \phi_i \loc r(x_i) \,\land \bigland_{\substack{i,j \in I_l\\i < j}}\! d(x_i,x_j) > 2r \Big)
\end{align*}
and $i_l \coloneqq \min I_l$ is a representative of the class $I_l$.
Essentially, line $(1)$ covers the case where $\tx$ is scattered and line (2) covers all cases where some variables are close (the idea is that $x_1,\dots,x_k$ are set to the representatives of the partition classes $I_1,\dots,I_k$).

We first argue that $\Psi$ can be written as local sentence.
The left subformula in (1) is already a basic local sentence, since $\biglor_{j=1}^m \phi_j \loc r(x_i)$ is a local formula around $x_i$.
The quantification over $\tx_{-j}$ and over $(x_1,\dots,x_k)$ both use less than $m$ variables,
so by induction they can be expressed as local sentences.
It thus remains to prove $\Phi \meq \Psi$.

\bigskip
We first show that $\Phi \mle \Psi$.
Fix a $K$-interpretation $\pi$ over universe $A$ and let $\ta = (a_1,\dots,a_m) \subseteq A$ be a tuple at which the maximum in $\Phi$ is reached.
If $\ta$ is $4r$-scattered, then by using $\ta$ as witness in $(1)$, we obtain
\begin{align*}
    \pi \ext \Psi \;\ge\; \pi \ext {(1)} \;\ge\;
    \min\big(\pi \ext {\phi_1 \loc r (a_1)}, \dots, \pi \ext {\phi_m \loc r (a_m)}\big) =
    \pi \ext \Phi.
\end{align*}

Otherwise, $\ta$ is $r$-scattered, but not $4r$-scattered.
By the \ClusteringLemma, there is a partition $P = \{I_1,\dots,I_k\}$ such that $\ta$ is in configuration $(P,4r)$.
Recall that the representatives $a_{i_l}$ (with $i_l = \min I_l$) satisfy $d(a_{i_l},a_{i_{l'}}) > 4 \cdot 5^{n-k} 4r$, so they are $(2 \cdot 5^{n-k} 4r)$-scattered.
In particular $P \neq \{\{1\},\dots,\{m\}\}$, because in this case $\ta$ would be $4r$-scattered.
By using the representatives $(a_{i_1},\dots,a_{i_k})$ as witness in (2) and the remaining entries of $\ta$ as witnesses in the $\phi_{P,l}$, we obtain:
\begin{align*}
    \pi \ext \Psi \;\ge\;
    \pi \ext {(2)} \;&\ge\;
    \pi \EXT {\bigland_{l=1}^k \phi_{P,l}\loc {5^{m-k}4r}(a_{i_l})} \\ \;&\ge\;
    \pi \EXT {\bigland_{l=1}^k \Big(\bigland_{i \in I_l} \phi_i \loc r(a_i) \,\land \bigland_{\substack{i,j \in I_l\\i < j}}\! d(a_i,a_j) > 2r\Big)} \\ \;&\overset{\mathclap{(*)}}=\;
    \prod_{i=1}^{m} \pi \ext {\phi_i \loc r (a_i)} \;=\; \pi \ext \Phi.
\end{align*}
For $(*)$, recall that $\ta$ is $r$-scattered, so the atoms $d(a_i,a_j) > 2r$ all evaluate to $1$.

\bigskip
We now prove that $\Psi \mle \Phi$.
We again fix a $K$-interpretation $\pi$ and first observe that $\pi \ext {(2)} \le \pi \ext \Phi$.
To see this, note that the $m$ elements quantified in $(2)$ and in all $\phi_{P,l}$ together form an $r$-scattered tuple.
Indeed, the $k$-tuple of representatives in $(2)$ is even $(2 \cdot 5^{m-k} \cdot 4r)$-scattered.
Hence for $l \neq l'$, elements locally quantified in $\phi_{P,l}(x_l)$ have distance at least $(4 \cdot 5^{m-k}4r) - 2 \cdot (5^{m-l}4r-r) \ge 2r$ to those locally quantified in $\phi_{P,l'}(x_{l'})$.
Within one $\phi_{P,l}$, we explicitly assert that the locally quantified elements have distance at least $2r$ among each other and to the representative.
We can thus use the quantified elements as $r$-scattered witness in $\Phi$.

It remains to show $\pi \ext {(1)} \le \pi \ext \Phi$.
Let $\tb = (b_1,\dots,b_m)$ be an $m$-tuple at which the maximum in $\Esc {4r} {\tx} \bigland_{i=1}^m \biglor_{j=1}^m \phi_j \loc r(x_i)$ is reached.
Let $\ta_j = (a_{1,j},\dots,a_{j-1,j},a_{j+1,j},\dots,a_{m,j})$ be an $(m-1)$-tuple at which the maximum in $\Esc {4r} {\tx_{-j}} \, \bigland_{i \neq j} \phi_i \loc r (x_i)$ is reached, for each $j$.
We say that an element $d \in A$ is \emph{associated to $\phi_i$} if
$\pi \ext {(1)} \le \pi \ext {\phi_i \loc r (d)}$.
Thus, for all $i \neq j$, the element $a_{i,j}$ is associated to $\phi_i$, and all $b_j$ are associated to some $\phi_i$. To end the proof, it suffices to find an $r$-scattered tuple $\tup d = (d_1,\dots,d_m)$ such that $d_i$ is associated to $\phi_i$ for all $i$, as we can then use $\td$ as witness in $\Phi$.

\begin{claim*}
There is a $2r$-scattered tuple $\tc = (c_1,\dots,c_m)$ such that $c_i$ is associated to $\phi_i$ for $i \ge 2$, and $c_1$ is associated to $\phi_{i_0}$ for some $i_0$.
\end{claim*}

\begin{claimproof}
Since $\tb$ is $4r-$scattered, the $m$ balls $\ball{2r}{b_j}$ are disjoint.
One of them, say $\ball{4r}{b_{j_0}}$, contains none of the $m-1$ entries of $\ta_1$.
Set $i_0$ so that $b_{j_0}$ is associated to $\phi_{i_0}$.
We can then take $\tc = (b_{j_0}, \ta_1)$.
\end{claimproof}

Let $\tc$ be such a tuple and set $i_0$ accordingly.
Observe that we can simply take $\td = \tc$ if $i_0 = 1$, so assume that $i_0 \neq 1$.
We construct a tuple $\td$ with the desired property from $\tc$ and $\ta_{i_0}$ by the following algorithm.

\begin{algorithm}[H]
    initialise $\tup d_{-i_0} = \tup a_{i_0}$ and $d_{i_0} = \bot$\;
    \While{$d_{i_0} = \bot$}{
        choose $i^*$ such that $d_{i^*} \neq c_{i*}$ and for all $i$: $d_i \notin \ball{2r}{c_{i^*}}$ if defined\;
        \eIf{$i^* = 1$ or $i^* = i_0$}{
            set $d_{i_0} = c_{i^*}$\;
        }{
            set $d_{i^*} = c_{i^*}$\;
        }
    }
\end{algorithm}

\begin{claim*}
This algorithm has the following invariant: $\td$ is $r$-scattered and for each $i$, the entry $d_i$ is either undefined or associated to $\phi_i$.
\end{claim*}
\begin{claimproof}
Whenever we update an entry $d_i$, the new value is associated to $\phi_i$ by choice of $\tc$ and $i_0$.
For the first part, recall that $\tup a_{i_0}$ is $r$-scattered.
In each iteration, we choose $i^*$ so that $c_{i^*}$ is not close to any entry of $\td$, hence $\td$ remains $r$-scattered after each update.
\end{claimproof}

The algorithm clearly terminates:
the if-case defines $d_{i_0}$ (thus ending the loop), and the else-case can only be executed once per entry of $\td$, as we require $d_{i^*} \neq c_{i^*}$.
It remains to show that while $d_{i_0} = \bot$, an $i^*$ with the desired property always exists.
There are $m-1$ defined entries of $\td$.
Let $n = |\{j: d_j \in \tc\}|$ be the number of entries contained in $\tc$.
The remaining $(m-1)-n$ defined entries are from $\ta_{i_0}$.
Recall that $\tc$ is $2r$-scattered, hence the balls $(\ball {2r} {c_j})_{1 \le j \le n}$ are disjoint.
There are $m-n$ entries of $\tc$ that do not occur in $\td$, hence there is at least one entry $c_j$ such that $\ball {2r} {c_j}$ does not contain any of the $(m-1)-n$ entries from $\ta_{i_0}$.
We can then choose $i^* = j$.

This proves the algorithm correct, so there is always an $r$-scattered tuple $\td$ such that $d_i$ is associated to $\phi_i$ for all $i$.
We conclude that $\pi \ext {(1)} \le \pi \ext \Phi$ and hence $\Psi \mle \Phi$.
\end{proof}

\subsection{Gaifman Normal Form for $\E^*$-Sentences}

\begin{proposition}[Step \enumref{2}]
\label{thmGaifmanQuantifiers}
Every sentence of the form $\E x_1 \dots \E x_m \ \phi \loc r (x_1,\dots,x_m)$ is minmax-equivalent to an existential local sentence.
\end{proposition}

We proceed with step \enumref{2} of the proof.
The goal is to replace $\E \tx$ by a scattered quantifier $\Esc R {\tx}$, and to split $\phi \loc r(\tx)$ into separate formulae $\phi_i \loc R (x_i)$ local only around a single $x_i$ (the previous lemma allows us to use different $\phi_i$ for each $x_i$).
We achieve this by using a disjunction over all partitions $P \in \Part(m)$;
for a given partition $P = \{I_1,\dots,I_k\}$ we can do a scattered quantification of the representatives $x_{i_1},\dots,x_{i_k}$, as they are far apart.
The elements of each class $I_l$ can then be locally quantified around $x_{i_l}$.
Applying the \SeparationLemma{} allows us to split $\phi \loc r(\tx)$ into separate local formulae for each class.

\begin{proof}
By the \ClusteringLemma{}, we have $\Phi \meq \biglor_{P \in \Part(m)} \Phi_P$ for
\begin{align*}
    \Phi_P \coloneqq{} &\Esc {(2 \cdot 5^{m-k}r)} {x_{i_1},\dots,x_{i_k}} \\
        &\quad \Eball {(x_i)_{i \in I_1 \backslash \{i_1\}}} {5^{m-k}r - r} {x_{i_1}} \\
        &\quad\quad \dots \\
        &\quad\quad\quad \Eball {(x_i)_{i \in I_k \backslash \{i_k\}}} {5^{m-k}r - r} {x_{i_k}} \ \
        \phi \loc r (x_1,\dots,x_m),
\end{align*}
where $P = \{I_1,\dots,I_k\}$ with representatives $i_l \coloneqq \min I_l$, as usual.
Indeed, $\Phi_P$ evaluates to the maximum of $\pi \ext {\phi\loc{r}(\ta)}$ over all tuples $\ta$ in configuration $(P,r)$, and any tuple is in configuration $(P,r)$ for some partition $P$ by the lemma.

We fix a partition $P$ and show that $\Phi_P$ is minmax-equivalent to a local sentence; this implies the claim.
Notice that $\Phi_P$ is not yet a basic local sentence, as the inner formula depends on all quantified variables $x_{i_1},\dots,x_{i_k}$ at once, instead of being applied to all variables individually.
This can be achieved by the \SeparationLemma{}, after some preparation.
Let $X_i$ be the set of bound variables in $\phi \loc r(x_1,\dots,x_m)$ that are locally quantified around $x_i$.
We say that a literal $R\tup y$ (or an equality $y_1 = y_2$) is \emph{split by $P$} if $\tup y$ contains variables from $\{x_i\} \cup X_i$ and $\{x_j\} \cup X_j$ for some $i \in I_l$ and $j \in I_{l'}$ with $l \neq l'$, so $i$ and $j$ are in different partition classes.

Now let $\hat\phi \loc r (x_1,\dots,x_m)$ result from $\phi \loc r (x_1,\dots,x_m)$ by replacing all positive literals (including equalities) that are split by $P$ with \emph{false} (i.e., $x_1 \neq x_1$) and all negative literals split by $P$ with \emph{true} ($x_1 = x_1$).
Let further $\hat\Phi_P$ result from $\Phi_P$ by replacing $\phi \loc r$ with $\hat\phi \loc r$.
Then $\hat \Phi_P \meq \Phi_P$.
Indeed, consider a $K$-interpretation $\pi$ and two variables $y_i \in \{x_i\} \cup X_i$ and $y_j \in \{x_j\} \cup X_j$ with $i \in I_l$, $j \in I_{l'}$, $l \neq l'$.
Since $\phi$ is $r$-local, $y_i$ will be instantiated by an element at distance at most $r$ to the instantiation of $x_i$, which in turn has distance at most $5^{m-k}r-r$ to the element assumed by the representative $x_{i_l}$ of $I_l$.
Hence (the instantiations of) $y_i$ and $y_j$ both have distance at most $5^{m-k}r$ to their representatives $x_{i_l}$ and $x_{i_{l'}}$, which are at least $4 \cdot 5^{m-k}r$ apart.
Thus, $y_i$ and $y_j$ are at a distance greater than $2 \cdot 5^{m-k}r \ge 2$, so $\pi$ evaluates all positive literals and equalities involving $y_i$ and $y_j$ to $0$, and negative ones to $1$.

By the \SeparationLemma{}, $\hat \phi_P$ is minmax-equivalent to a positive Boolean combination of local formulae $\theta_l \loc r$ around $\{x_i: i\in I_l\}$ for each $l$.
By putting the Boolean combination in DNF and using the fact that $\E$ distributes over $\lor$ (cf.\ \cref{thmAxioms}), $\hat\Phi_P$ is minmax-equivalent to a disjunction over sentences of the form
\begin{align*}
    &\Esc {(2 \cdot 5^{m-k}r)} {x_{i_1},\dots,x_{i_k}} \\
        &\quad \Eball {(x_i)_{i \in I_1 \backslash \{i_1\}}} {5^{m-k}r - r} {x_{i_1}} \\
        &\quad\quad \dots \\
        &\quad\quad\quad \Eball {(x_i)_{i \in I_k \backslash \{i_k\}}} {5^{m-k}r - r} {x_{i_k}} \ 
        \vphantom{\bigland_{l=1}}\smash{\bigland_{l=1}^{k}} \theta_l \loc r (\{x_i : i \in I_l\}).
\end{align*}
Notice that each conjunct $\theta_l \loc r (\{x_i : i \in I_l\})$ is independent of the variables $x_i$ for $i \notin I_l$.
We can thus write such a sentence equivalently as a local sentence (using \cref{thmGaifmanGeneralizedBasic}):
\[
    \Esc {(2 \cdot 5^{m-k}r)} {x_{i_1},\dots,x_{i_k}} \ \bigland_{l=1}^k \Tilde \theta_l^{(5^{m-k}r)}(x_{i_l}),
\]
with $\displaystyle\Tilde \theta_l \loc {5^{m-k}r} (x_{i_l}) = \Eball {(x_i)_{i \in I_l \backslash \{i_l\}}} {5^{m-k}r - r} {x_{i_l}} \ \theta_l \loc r(\{x_i : i \in I_l\})$.
\end{proof}

\medskip
\begin{remark}
\label{remarkExistentialGaifman}
\Cref{thmGaifmanQuantifiers} implies that every existential first-order sentence is minmax-equivalent to positive Boolean combination of existential basic local sentences.
This is very similar to a result of Grohe and Wöhrle \cite{GroheWoe04} for standard Boolean semantics.
In fact, their statement reads exactly the same.

However, their definition of \emph{existential (basic) local sentences} is more strict, as the inner local formula has to be purely existential.
This is not the case in our construction, since \cref{thmGaifmanGeneralizedBasic} adds distance formulae $d(x,y) > 2r$ within a local formula, which are abbreviations for universal quantifiers.
\end{remark}

\subsection{Gaifman Normal Form for $\E^*\A^*$-Sentences}

It remains to show \enumref{3}, which is the most challenging step of the proof.
The overall structure of our proof is similar to Gaifman's original proof \cite{Gaifman82}, but each individual step is much more involved due to the stronger notion of equivalence we consider.
Our proof proceeds as follows.

\begin{proof}[Proof outline]
Starting from a $\E^*\A^*$-sentence of the form
\[
    \Psi_1 = \E \tup y \A \tx \, \phi \loc r(\tup y, \tx),
\]
we first split $\A \tx$ into quantification over elements close to $\ty$, and elements far from $\ty$.
By the \SeparationLemma{}, we can split $\phi \loc r$ into a local formula around the close elements, and a local formula around the far elements.
It follows that $\Psi_1$ is minmax-equivalent to a positive Boolean combination of sentences of the form
\[
    \Psi_2 = \E \tup y \big( \phi\loc{R}(\tup y) \quad\land\quad \A x_{1} \notin \ball{r_1}{\tup y}\dots\A x_{N} \notin \ball{r_N}{\tup y} \ \psi^{(r)}(\tx) \big),
\]
for a sufficiently large number $N$ and radii $r_1,\dots,r_N,R$.
Notice that here we use \emph{outside quantifiers} of the form $\Anotball x r y$ with the obvious semantics.
The rest of the proof is concerned with turning the outside quantifiers into a local sentence.

\medskip
As a first (non-trivial) step, we show that we can restrict the outside quantifiers to a single one.
That is, each sentence $\Psi_2$ is minmax-equivalent to a disjunction over sentences of the form
\[
    \Psi_3 = \E \tup w \big( \phi \loc {r} (\tup w) \;\land\; \Anotball z {r'} {\tup w} \; \psi \loc {r'} (z) \big).
\]
By similar reasoning as in step \enumref{2}, we can then rewrite $\E \tup u$ as scattered quantification.
That is, $\Psi_3$ is minmax-equivalent to a positive Boolean combination of sentences of the form
\[
    \Psi_4 = \Esc {r}{u_1,\dots,u_M} \Big( \bigland_{i=1}^M \phi_i \loc{r}(u_i) \;\land\; \Anotball{z}{r'}{\tup u} \; \psi\loc{r''}(z) \Big).
\]
For sentences of the form $\Psi_4$, we can directly specify a minmax-equivalent sentence, that can be written as a local sentence by the previous results, without outside quantifiers.
\end{proof}

\subsubsection{Close and far}

In the first step, we split the quantification over $\tx$ into those $\tx$ that are close to $\ty$, and those far from $\ty$.

\begin{proposition}
\label{thmGaifmanAlternation1}
Sentences of the form $\E \tup y \A \tx \, \phi \loc r(\tup y, \tx)$ are equivalent to positive Boolean combinations of sentences of the form $\E \tup y (\phi\loc{R}(\tup y) \land \Anotball{x_{1}}{r_1}{\tup y}\dots\Anotball{x_{N}}{r_N}{\tup y} \ \psi^{(r)}(\tx))$.
\end{proposition}

\begin{proof}[Proof sketch]
We first sketch the idea for a tuple $\tx$ of length $2$.
In order to split $\phi \loc r(\tup y, \tx)$ into a formula local around $\tup y$ and one that is far from $\tup y$, we replace each variable $x_i$ with two variables $x_i^{\textit{in}}$ and $x_i^{\textit{out}}$, using the equivalence
\[
    \A x \ \theta(x,\tup z) \quad\meq\quad
    \Aball{x^{\textit{in}\,}}{r_0}{\tup z} \ \theta(x^{\textit{in}},\tup z) \;\land\;
    \Anotball{x^{\textit{out}\,}}{r_0}{\tup z} \ \theta(x^{\textit{out}},\tup z).
\]
For instance, we would write $\A x_1 \A x_2 \ \phi \loc r(\tup y,x_1,x_2)$ equivalently as
{\small
\begin{align*}
    \Aball{x_1^{\textit{in}}}{r_0}{\tup y} &\Big( \Aball{x_2^{\textit{in}}}{r_0}{\tup y} \ \phi \loc r(\ty, x_1^{\textit{in}},x_2^{\textit{in}})
    \;\land\;
    \Anotball {x_2^{\textit{out}}}{r_0}{\tup y} \ \phi \loc r(\ty, x_1^{\textit{in}},x_2^{\textit{out}}) \Big) \\
    {}\land \Anotball {x_1^{\textit{out}}}{r_0}{\tup y} &\Big( \Aball{x_2^{\textit{in}}}{r_0}{\tup y} \ \phi \loc r(\ty, x_1^{\textit{out}},x_2^{\textit{in}})
    \;\land\;
    \Anotball {x_2^{\textit{out}}}{r_0}{\tup y} \ \phi \loc r(\ty, x_1^{\textit{out}},x_2^{\textit{out}}) \Big).
\end{align*}}%

We want to apply the \SeparationLemma{} to achieve the desired form. 
To this end, we choose $r_0 > 2r+1$ so that variables locally quantified around $x_i^{\textit{out}}$ are not close to variables locally quantified around $\ty$ in $\phi \loc r$.
This permits us to replace all atoms (and equalities) involving both variables local around $x_i^{\textit{out}}$ and around $\ty$ by \emph{false} (i.e., by $y_1 \neq y_1$).
However, we still cannot apply the Separation Lemma to $\{\ty,x_1^{\textit{in}},x_2^{\textit{in}}\}$ and $\{x_1^{\textit{out}},x_2^{\textit{out}}\}$, since, e.g., $x_1^{\textit{in}}$ and $x_2^{\textit{out}}$ may be close, and may both appear in a literal that cannot be removed.
We solve this issue by taking different radii for $x_i$ depending on the order of quantification, so that variables $x_i^{\textit{in}}$ and $x_j^{\textit{out}}$ for $i \neq j$ are always at distance $\ge r_0$.
For instance,
{\small
\begin{align*}
    \Aball{x_1^{\textit{in}}}{2r_0}{\tup y} &\Big( \Aball{x_2^{\textit{in,in}}}{3r_0}{\tup y} \ \phi \loc r(\ty,x_1^{\textit{in}},x_2^{\textit{in,in}})
    \;\land\;
    \Anotball{x_2^{\textit{in,out}}}{3r_0}{\tup y} \ \phi \loc r(\ty,x_1^{\textit{in}},x_2^{\textit{in,out}}) \Big) \\
    \!\!\!{}\land \Anotball{x_1^{\textit{out}}}{2r_0}{\tup y} &\Big( \Aball{x_2^{\textit{out,in}}}{r_0}{\ty} \ \phi \loc r(\ty,x_1^{\textit{out}},x_2^{\textit{out,in}})
    \;\land\;
    \Anotball{x_2^{\textit{out,out}}}{r_0}{\tup y} \ \phi \loc r(\ty,x_1^{\textit{out}},x_2^{\textit{out,out}}) \Big).
\end{align*}}%
In this example, $x_1^{\textit{in}}$ and $x_2^{in,out}$ are at distance $ \ge r_0$, and so are $x_1^{\textit{out}}$ and $x_2^{out,in}$.
\end{proof}

\begin{proof}
To generalise this idea to $n$ variables $\tx$, we annotate variables by sequences $s \in \{-1,1\}^n$.
That is, for each sequence $s$ we introduce variables $x_1^{s}, x_2^{s},\dots,x_n^{s}$ with the intention that $s_i = 1$ indicates that $x_i^s \in \ball {r_{i,s}} {\ty}$, whereas $s_i = -1$ means that $x_i^s \notin \ball {r_{i,s}} {\ty}$.
We set the radius as $r_{i,s} = (2^{n-1} + \sum_{j=1}^{i-1}s_j 2^{n-1-j})r_0$, with $r_0 = 4r$ (so that $r_0 > 2r+1$).
With those notations, we have:
\begin{align*}
    \A \tx \ \phi \loc r(\tup y,\tx) \;\meq
    \bigland_{s \in \{-1,1\}^n} Q_{s_1}x_1^{s} \dots Q_{s_i}x_i^{s} \dots Q_{s_n}x_n^s \ \phi \loc r(\tup y, x_1^{s},\dots,x_n^s),
\end{align*}
where $Q_1 x_i = \Aball{x_i}{r_{i,s}}{\ty}$ and $Q_{-1} x_i = \Anotball {x_i}{r_{i,s}}{\ty}$.
It remains to verify that for each sequence $s$, variables $x_i^s$ with $s_i=1$ and $x_j^s$ with $s_j = -1$ are at distance $\ge r_0$, as intended.
First assume $i < j$.
Then
\begin{align*}
    r_{j,s} =
    r_{i,s} + \big(s_i 2^{n-1-i} + \sum_{\mathclap{k = i+1}}^{j-1}s_k2^{n-1-k}\big)r_0 \ge
    r_{i,s} + \big(2^{n-1-i} - \sum_{\mathclap{k = i+1}}^{j-1}2^{n-1-k}\big)r_0 \ge
    r_{i,s} + r_0,
\end{align*}
so $x_j^s$ is outside of a $(r_{i,s}+r_0)$-ball around $\ty$, and hence at distance $\ge r_0$ to $x_i^s$.
If $j < i$, a similar computation gives $r_{i,s} \le r_{j,s} - r_0$, so again $x_j^s$ is at distance $\ge r_0$ to $x_i^s$.

Fixing a sequence $s$, we write $\tx^s_{\text{in}} = \{ x_i^s \mid s_i = 1\}$ and $\tx^s_{\text{out}} = \{ x_i^s \mid s_i = -1\}$.
We can remove those literals in $\phi \loc r(\ty,x_1^s,\dots,x_n^s)$ having both a variable close to $\ty$, that is a variable local around $\ty \cup \tx^s_{\text{in}}$, and a variable far from $\ty$, i.e., local around $\tx^s_{\text{out}}$.
By the \SeparationLemma{}, we obtain an equivalent positive Boolean combination of formulae local around $\ty \cup \tx^s_{\text{in}}$ and formulae local around $\tx^s_{\text{out}}$.
Writing the Boolean combination in CNF and simplifying (using \cref{thmAxioms}), it follows that $\E \ty \A \tx \; \phi \loc r(\ty,\tx)$ is minmax-equivalent to
\begin{align*}
    &\E\ty \bigland_s \Big(
        \big(\Aball {x_i^s} {r_{i,s}} {\ty} \big)_{x_i^s \in \tx^s_{\text{in}}} \;
        \big(\Anotball {x_i^s} {r_{i,s}} {\ty} \big)_{x_i^s \in \tx^s_{\text{out}}} \;
        \bigland_j (\theta_{j,s} \loc r (\ty,\tx^s_{\text{in}}) \,\lor\, \psi_{j,s} \loc r (\tx^s_{\text{out}})) \Big) \\
    \meq{} &\E\ty \bigland_s \bigland_j \Big( \tilde\theta_{j,s} \loc R (\ty) \,\lor\, 
        \big(\Anotball {x_i^s} {r_{i,s}} {\ty} \big)_{x_i^s \in \tx^s_{\text{out}}} \psi_{j,s} \loc r (\tx^s_{\text{out}}) \Big),
\end{align*}
with $R = 2^n r_0$ so that $\tx^s_{\text{in}}$ can be locally quantified within $\tilde\theta_j \loc R(\ty)$.
By putting this in DNF, we see that $\E \ty \A \tx \; \phi \loc r(\ty,\tx)$ is minmax-equivalent to a disjunction over sentences of the form
\[
    \E \ty \Big(\theta \loc R (\ty) \,\land\,
    \big(\Anotball {x_i} {r_i} {\ty} \big)_{1 \le i \le N} \ \psi \loc r (x_1,\dots,x_N) \Big),
\]
where $N \le n \cdot 2^n$ to account for the quantification over $\tx^s_{\text{out}}$ for all $2^n$ choices of $s$,
and $r_1,\dots,r_N,R \le 2^n r_0$ due to the choice of the radii $r_{i,s}$.
\end{proof}

\subsubsection{Restricting to one outside quantifier}

In the next step, we prove that the outside quantifiers $\Anotball{x_{1}}{r_1}{\tup y}\dots\Anotball{x_{N}}{r_N}{\tup y} \ \psi^{(r)}(\tx)$ can be collapsed into a single outside quantifier.
We achieve this by applying the analogue of \cref{thmGaifmanQuantifiers} for universal sentences (with the help of the \AbstractionLemma{}).
To turn the result into an existential sentence (including a single outside quantifier), the following auxiliary lemma is crucial.

\begin{lemma}
\label{thmGaifmanAlternationUnivToExist}
Let $\Theta = \Asc {r} {x_1,\dots,x_n} \biglor_{i=1}^n \phi \loc r(x_i)$ be a universal basic local sentence. Then $\Theta \meq \Psi$ for
\begin{align*}
    \Psi \coloneqq \E v_1 \dots \E v_{n-1} \Big(&\Aball{x_1}{2r}{\tup v} \dots \Aball{x_n}{2r}{\tup v} \big( \biglor_{i<j} d(x_i,x_j) \le 2r \lor \biglor_{i=1}^n \phi \loc r(x_i)\big) \\
    &{}\land \Anotball{x}{2r}{\tup v} \ \phi \loc r(x) \Big).
\end{align*}
\end{lemma}

\begin{proof}
Let $\pi$ be a $K$-interpretation over universe $A$.
We first prove $\pi \ext {\Psi} \le \pi \ext {\Theta}$.
Recall that $\pi \ext {\Theta}$ is the minimum of $\max_i \pi \ext {\phi \loc r(a_i)}$ over all $r$-scattered tuples $\ta$.
Thus, it suffices to show that for all $r$-scattered tuples $\ta$, there exists $i$ such that $\pi \ext {\Psi} \le \pi \ext {\phi \loc r(a_i)}$.

Let $\ta$ be any $r$-scattered $n$-tuple and let $\tb$ be the $(n-1)$-tuple at which the maximum in $\Psi$ is reached.
If $a_i \notin \ball{2r}{\tb}$ for some $i$, then $\pi \ext {\Anotball{x}{2r}{\tb} \ \phi \loc r(x)} \le \pi \ext {\phi \loc r(a_i)}$.
Otherwise, choosing $\ta$ as witness for $x_1,\dots,x_n$ gives
\[
    \pi \Ext {\biglor_{i<j} d(a_i,a_j) \le 2r \lor \biglor_{i=1}^n \phi \loc r(a_i)} \;=\;
    \pi \Ext {\biglor_{i=1}^n \phi \loc r(a_i)} \;\le\;
    \pi \ext {\phi \loc r(a_{i_0})},
\]
for an index $i_0$ at which $\max_i \pi \ext {\phi \loc r(a_i)}$ is reached.
In both cases, we have $\pi \ext \Psi \le \pi \ext {\phi \loc r(a_i)}$ for some $i$, and hence $\pi \ext \Psi \le \pi \ext \Theta$.

\medskip
It remains to prove $\pi \ext {\Psi} \ge \pi \ext {\Theta}$.
We define $A' = \{ c \in A \mid$ $c \in \ta$ for an $r$-scattered tuple $\ta$ with $\pi \ext {\phi \loc r (c)} = \max_i \pi \ext {\phi \loc r (a_i)} \}$, so that for every $r$-scattered tuple $\ta$, its ``best'' entries (the ones maximizing $\phi$) are contained in $A'$.
Now let $\tb$ be an $r$-scattered tuple of maximal length with $b_i \notin A'$ for all $i$.
Clearly $|\tb| \le n-1$ by construction of $A'$.
If $|\tb|=k < n-1$, we extend $\tb$ to length $n-1$ by duplicating the last element $b_k$ (the resulting tuple is no longer $r$-scattered).
Using the resulting tuple $\tb$ as witness for $\tup v$, we claim that
\[
    \pi \ext {\Anotball{x}{2r}{\tb} \ \phi \loc r(x)} \;\ge\; \pi \ext \Theta.
\]
Indeed, consider any element $d \notin \ball{2r}{\tb}$.
By maximality of $\tb$, we must have $d \in A'$.
Hence $d$ is the ``best'' entry of an $r$-scattered tuple $\ta$, so $\pi \ext {\phi \loc r(d)} \ge \pi \ext \Theta$.

Moreover, observe that for any tuple $\tb$,
\[
    \pi \ext {\Aball{x_1}{2r}{\tb} \dots \Aball{x_n}{2r}{\tb} \big(\biglor_{i<j} d(x_i,x_j) \le 2r \lor \biglor_{i=1}^n \phi \loc r(x_i)\big)}
    \;\ge\;
    \pi \ext {\Theta}.
\]
To see this, let $\tc$ be the assignment to $\tx$ at which the minimum is reached.
Then $\tc$ is either $r$-scattered and thus a potential witness for $\Theta$, or the left-hand side evaluates to the greatest element $1$ due to $d(x_i,x_j) \le 2r$.
Together, we have shown $\pi \ext {\Psi} \ge \pi \ext {\Theta}$.
\end{proof}

We further need the following observation to combine several outside quantifiers (originating from different basic local sentences).
Intuitively, we can simply enlarge the radius to a sufficient size $R$, and add local formulae that take care of everything closer than $R$.

\begin{lemma}
\label{thmNotballCombine}
Let $\tx_1,\dots,\tx_n,\ty$ be variable tuples.
$\bigland_{i=1}^n \Anotball z {r_i} {\tup x_i} \, \phi_i \loc {r_i} (z \mid \ty)$ is minmax-equivalent to a formula of the following form, for any $R \ge \max_i r_i$,
\[
    \Anotball z R {\tx_1,\dots,\tx_n} \, \tilde\phi \loc {\max_i r_i} (z \mid \ty) \;\land\; \theta \loc{2R} (\tx_1,\dots,\tx_n \mid \ty).
\]
\end{lemma}
\begin{proof}
Simple case distinction: each $z \notin \ball {r_i} {\tup x_i}$ either satisfies also $z \notin \ball {R} {\tup x_1,\dots,\tup x_n}$,
or we have $z \in \ball R {\tup x_i}$ and $d(z,\tup x_i) > r_i$.
Hence the minmax-equivalent formula is
\begin{align*}
    &\Anotball z R {\tup x_1,\dots,\tup x_n} \underbrace{\bigland_{i=1}^n \phi_i \loc {r_i} (z \mid \ty)}_{\tilde\phi\loc {\max_i r_i}} \;\land\;
    \underbrace{\bigland_{i=1}^n \Aball z R {\tup x_i} \big( \bigland_{j=1}^n (d(z,\tup x_j) \le r_j \,\lor\, \phi_j \loc {r_j}(z \mid \ty))}_{\theta\loc{2R}, \text{ since $2R \ge R + \max_i r_i$}} \big). \qedhere
\end{align*}
\end{proof}
\smallskip %hack

\begin{proposition}
\label{thmGaifmanAlternation2}
Sentences of the form
$\E \ty \big(\phi \loc{R}(\ty) \land \Anotball {x_1}{r_1}{\ty} \dots \Anotball{x_n}{r_n}{\ty} \; \psi \loc{R}(\tx)\big)$
are minmax-equivalent to positive Boolean combinations of sentences of the form
\[
    \E \tup w \big( \hat\phi \loc {r} (\tup w) \;\land\; \Anotball z {r'} {\tup w} \, \hat\psi \loc {r'} (z) \big).
\]
\end{proposition}

\begin{proof}
Let $\Phi = \E \ty \big(\phi \loc{R}(\ty) \land \Anotball {x_1}{r_1}{\ty} \dots \Anotball{x_n}{r_n}{\ty} \ \psi \loc{R}(\tx)\big)$ over signature $\tau$.
We can replace quantifiers $\Anotball {x_i}{r_i}{\ty}$ using distance formulae $d(x_i,\tup y) \le r_i$.
Thus,
\[
    \Phi \meq \E \ty \Big( \phi \loc R(\ty) \;\land\; \A x_1 \dots \A x_n \big( \psi \loc R(\tx) \lor \biglor_{i=1}^n d(x_i,\overline y) \le r_i \big) \Big).
\]

We would like to apply (the dual of) \cref{thmGaifmanQuantifiers} to show that $\A x_1 \dots \A x_n (\dots)$ is a local sentence.
However, it is not exactly in the required form due to the additional free variables $\ty$ in $d(x_i,\ty) \le r_i$.
To resolve this problem, we apply the \AbstractionLemma{} to abstract subformulae $d(x_i,\ty) \le r_i$ by unary atoms $D_{r_i} x_i$ with fresh relation symbols:
\[
    \Psi \coloneqq \A x_1 \dots \A x_n \big( \psi \loc R(\tx) \lor \biglor_{i=1}^n D_{r_i} x_i \big),
    \qquad
    \text{over $\tau' = \tau \cup \{ D_{r_i} : 1 \le i \le n \}$}.
\]
We now apply (the dual of) \cref{thmGaifmanQuantifiers} followed by \cref{thmGaifmanAlternationUnivToExist}, implying that $\Psi$ is minmax-equivalent to a positive Boolean combination
$\Psi \meq \biglor_k \bigland_l \Psi_{kl}$, where each $\Psi_{kl}$ is a $\tau'$-sentence of the form (for some radius $r$):
\begin{align*}
   \E v_1 \dots \E v_{n-1} \Big(&\Aball{x_1}{2r}{\tup v} \dots \Aball{x_n}{2r}{\tup v} \big( \biglor_{i<j} d(x_i,x_j) \le 2r \lor \biglor_{i=1}^n \tilde\psi \loc r(x_i)\big) \tag{$\Psi.1$} \\
    &{}\land \Anotball{x}{2r}{\tup v} \ \tilde\psi \loc r(x) \Big). \tag{$\Psi.2$}
\end{align*}

Recall that each $\tilde\psi\loc r$ may contain the newly introduced symbols $D_{r_i}$.
Replacing these symbols by the original distance formula, the \AbstractionLemma{} gives
\begin{align*}
    \Phi &\meq \E \ty \Big(\phi \loc R (\ty) \;\land\; \biglor_k \bigland_l \Psi_{kl}[D_{r_i} x_i / d(x_i,\ty) \le r_i] \Big). \\
    &\meq \biglor_k \E \ty \bigland_l \Big( \phi \loc R (\ty) \;\land\; \Psi_{kl}[D_{r_i} x_i / d(x_i,\ty) \le r_i] \Big).
\end{align*}
Each $\Psi_{kl}$ existentially quantifies over a tuple $\tup v$ of variables.
By pulling all of these quantifiers to the front of $\Phi$, we can combine all $\Psi_{kl}[D_{r_i} x_i / d(x_i,\ty) \le r_i]$ into a conjunction of local formulae (line $(\Psi.1)$), which may include $\ty$ as free variable due to the substitution, and a conjunction of formulae of the form $\Anotball{x}{2r}{\tup v} \ \tilde\psi \loc r(x \mid \ty)$ (line $(\Psi.2)$).
By \cref{thmNotballCombine}, we can combine all of the outside quantifiers into a single one (with the help of another local formula).
Hence $\Phi$ is minmax-equivalent to a disjunction over sentences of the form
\[
    \E \tup w \big( \hat\phi \loc {\rho + 2\rho'} (\tup w) \;\land\; \Anotball z {\rho'} {\tup w} \, \hat\psi \loc {\rho''} (z \mid \ty) \big).
\]
Here, $\rho,\rho',\rho''$ are sufficiently large radii,
and $\tup w$ contains $\ty$ and all tuples $\tup v$ from the $\Psi_{kl}$.
We can further eliminate the occurrences of $\ty$ in $\Anotball z {\rho'} {\tup w} \, \hat\psi \loc {\rho''} (z \mid \ty)$.
Recall that $\ty$ may only occur in atoms of the form $d(x,\ty) \le r$ for some variable $x$ and some radius $r \le \rho''$.
Since all variables locally quantified in $\hat\psi \loc {\rho''} (z \mid \ty)$ are $\rho''$-local around $z$, we can choose $\rho' \gg \rho''$ sufficiently large (cf.\ \cref{thmNotballCombine}) so that $d(x,\ty) \le r$ will never hold.
We can thus replace all atoms $d(x,\ty) \le r$ by false (i.e., by $z \neq z$) to eliminate the occurrences of $\ty$.
To close the proof, we remark that $\hat\psi\loc{\rho''}$ is in particular $\rho'$-local.
\end{proof}

\subsubsection{Scattered quantification}

So far we have shown that it suffices to consider sentences of the form
\[
    \E x_1 \dots \E x_m \big( \phi \loc r (\tx) \;\land\; \Anotball z {r'} {\tx} \; \theta\loc {r'}(z) \big).
\]
The next step is to replace the existential quantifiers by scattered quantification.
The proof is very similar to the one for $\E^*$-sentences in \cref{thmGaifmanQuantifiers}, but we additionally need the following lemma to take care of the outside quantifier.

\begin{lemma}
\label{thmNotballPartition}
Let $P = \{I_1,\dots,I_k\} \in \Part(m)$ with representatives $i_l \coloneqq \min I_l$.
Let $\pi$ be a $K$-interpretation over universe $A$ (for a min-max semiring $K$).
Then
\[
    \pi \ext {\Anotball z r {\ta} \, \theta\loc r(z)} = \pi \ext {\psi (a_{i_1},\dots,a_{i_k})}
\]
for all tuples $\ta = (a_1,\dots,a_m) \subseteq A$ in configuration $(P,r)$, where $R = 5^{m-k}r$ and
\begin{align*}
    \psi(x_{i_1},\dots,x_{i_k}) \coloneqq{} &\Anotball z R {x_{i_1},\dots,x_{i_k}} \; \theta\loc r(z) \\
    &\land\; \bigland_{l=1}^k \Aball z R {x_{i_l}} \big( \biglor_{i \in I_l} d(z,x_i) \le r \,\lor\, \theta \loc r(z) \big).
\end{align*}
\end{lemma}

\begin{proof}
We use the following intuition.
Instead of quantifying over all elements that are at least $r$-far from $\ta$, we first quantify over those that are very far (i.e., $R$-far) from the representatives $a_{i_1},\dots,a_{i_k}$, and thus far from all elements $\ta$.
It remains to quantify over elements $z$ that are $R$-close to some representative, say $a_{i_l}$, but at least $r$-far from all elements $\ta$.
The crucial insight: it suffices to check that $z$ is $r$-far from all elements in the class of $a_{i_l}$, as elements in other classes are sufficiently far away from $a_{i_l}$ and hence also from $z$.

\medskip
Formally, let $\pi$, $\ta$, $P$, $R$ be as above.
We have $\pi \ext {\Anotball z r {\ta} \, \theta\loc r(z)} = \min_{c \notin \ballpi{r}{\ta}} \pi \ext {\theta\loc r(c)}$
and we prove that $\pi \ext {\psi (a_{i_1},\dots,a_{i_k})} = \min_{c \notin \ballpi{r}{\ta}} \pi \ext {\theta\loc r(c)}$ as well.
Notice that
\[
    \pi \Ext {\biglor_{i \in I_l} d(c,a_i) \le r \,\lor\, \theta \loc r(c)} = \begin{cases*}
        1, &if $c$ is $r$-close to $\{a_i \mid i \in I_l\}$, \\
        \pi \ext {\theta\loc r(c)}, &if $c$ is $r$-far from $\{a_i \mid i \in I_l\}$.
    \end{cases*}
\]
Hence $\pi \ext {\Aball z R {x_{i_l}} (\biglor_{i \in I_l} d(c,a_i) \le r \,\lor\, \theta \loc r(c))}$
evaluates to $\min_{c \in C_l} \pi \ext {\theta\loc r(c)}$, where
\begin{align*}
    C_l ={}
    &\{ c \in \ballpi R {a_{i_l}} : \text{$c$ is $r$-far from $\{a_i \mid i \in I_l\}$}\} \\
    \overset{\mathclap{(1)}}{=}{}
    &\{ c \in \ballpi R {a_{i_l}} : \text{$c$ is $r$-far from $\ta$}\}.
\end{align*}
To see why $(1)$ holds, consider an element $c \in \ballpi R {a_{i_l}}$ and an element $a_j$ with $j \in I_{l'}$ for a different class $l' \neq l$.
Due to the properties of the partition (see \cref{defConfiguration}), we have $d(a_{i_l},a_{i_l'}) > 4R$ and $d(a_j,a_{i_{l'}}) \le R-r$.
Since $d(c,a_{i_l}) \le R$, we have $d(c,a_j) > 2R+r$, so $c$ is indeed $r$-far from $a_j$.

Finally, $\pi \ext {\Anotball z R {a_{i_1},\dots,a_{i_k}} \; \theta\loc r(z)}$ evaluates to $\min_{c \in D} \pi \ext {\theta\loc r(c)}$, where
\begin{align*}
    D ={} 
    &\{ c : c \notin \ballpi R {a_{i_1},\dots,a_{i_k}} \} \\
    \overset{\mathclap{(2)}}{=}{}
    &\{ c : c \notin \ballpi R {a_{i_1},\dots,a_{i_k}} \text{ and $c$ is $r$-far from $\ta$} \}. 
\end{align*}
Here, $(2)$ holds by a similar argument: every $a_i$ with $i \in I_l$ is close to its representative, so $d(a_i,a_{i_l}) \le R-r$, whereas $d(c,a_{i_l}) > R$.
Hence $d(c,a_i) > r$ and $c$ is indeed $r$-far from $\ta$.
Since $D \cup \bigcup_{l=1}^k C_l = \{ c : \text{$c$ is $r$-far from $\ta$} \}$,
we have $\pi \ext {\psi (a_{i_1},\dots,a_{i_k})} = \min_{c \notin \ballpi{r}{\ta}} \pi \ext {\theta\loc r(c)}$.
\end{proof}

\begin{corollary}
\label{thmGaifmanQuantifiersNotball}
Every sentence of the form
$\E x_1 \dots \E x_m \big( \phi \loc r (\tx) \,\land\, \Anotball z {r'} {\tx} \, \theta\loc {r'}(z) \big)$
is minmax-equivalent to a positive Boolean combination of sentences of the form
\[
    \Esc{R}{x_1,\dots,x_k} \allowbreak \big( \bigland_{i=1}^k \tilde\phi_i \loc {R} (\tx) \;\land\; \Anotball z {R'} {\tx} \; \theta\loc {r'}(z) \big).
\]
\end{corollary}

\begin{proof}
We prove this corollary by closely following the proof of \cref{thmGaifmanQuantifiers} and applying the previous lemma to $\Anotball z {r'} {\tx} \, \theta\loc {r'}(z)$.
Let $\Phi$ be a sentence of the indicated form and assume w.l.o.g.\ that $r \ge r'$.
By applying the \Clustering{} and \Separation{} Lemmas, we obtain that $\Phi$ is minmax-equivalent to a disjunction over sentences of the form
\begin{align*}
    &\Esc {(2 \cdot 5^{m-k}r)} {x_{i_1},\dots,x_{i_k}} \\
        &\quad \Eball {(x_i)_{i \in I_1 \backslash \{i_1\}}} {5^{m-k}r - r} {x_{i_1}} \\
        &\quad\quad \dots \\
        &\quad\quad\quad \Eball {(x_i)_{i \in I_k \backslash \{i_k\}}} {5^{m-k}r - r} {x_{i_k}} \ 
        \Big( \vphantom{\bigland_{l=1}}\underbrace{\bigland_{l=1}^{k} \psi_l \loc r (\{x_i : i \in I_l\}) \;\land\; \Anotball z {r'} {\tx} \; \theta\loc{r'}(z)}_{\eta(\tx)} \Big).
\end{align*}
Observe that (the interpretation of) $\tx$ is guaranteed to be in configuration $(P,r)$ for a partition $P = \{I_1,\dots,I_k\}$ with $i_l = \min I_l$, due to the way the $x_i$ are quantified.
We set $R = 2 \cdot 5^{m-k}r$ and $R' = 5^{m-k}r'$.
By the previous lemma, we can equivalently replace $\eta(\tx)$ by
\begin{align*}
    \eta'(x_{i_1,\dots,x_{i_k}}) =
    &\bigland_{l=1}^{k} \Big(\psi_l \loc r (\{x_i : i \in I_l\}) \;\land\;
    \Aball z {R'} {x_{i_l}} \big( \biglor_{i \in I_l} d(z,x_i) \le r' \,\lor\, \theta \loc{r'}(z) \big) \Big) \tag{$*$}\\
    &\; \land \; \Anotball z {R'} {x_{i_1},\dots,x_{i_k}} \; \theta\loc{r'}(z).
\end{align*}
which depends only on the representatives.
The first line $(*)$ is a conjunction of $(R'+r)$-local formulae around $\{x_i : i \in I_l\}$, for each $l$ (recall that $r \ge r'$).
In particular, the $l$-th conjunct is independent of all $x_i$ with $i \notin I_l$, which allows us to pull the quantifiers inside.
It follows that $\Phi$ is minmax-equivalent to a disjunction of sentences of the desired form:
\[
    \Esc {R} {x_{i_1},\dots,x_{i_k}} \Big( \bigland_{l=1}^k \Tilde \phi\loc{R}(x_{i_l}) \;\land\; \Anotball z {R'} {x_{i_1},\dots,x_{i_k}} \; \theta\loc{r'}(z) \Big)
\]
with
\begin{align*}
    \Tilde \phi_l \loc {R} (x_{i_l}) = \Eball {(x_i)_{i \in I_l \backslash \{i_l\}}} {\frac 1 2 R-r} {x_{i_l}} \Big(
        &\psi_l \loc r (\{x_i : i \in I_l\}) \;\land\; \\
        &\Aball z {R'} {x_{i_l}} \big( \biglor_{i \in I_l} d(z,x_i) \le {r'} \,\lor\, \theta \loc{r'}(z) \big) \Big),
\end{align*}
which is indeed local around $x_{i_l}$ with radius $\le (\frac 1 2 R - r) + (R' + r) \le R$.
\end{proof}

\subsubsection{Eliminating the outside quantifier}

In the final step, we eliminate the outside quantifier by directly specifying a minmax-equivalent sentence composed of three parts, each of which can be rewritten as local sentence by our previous results.

\begin{proposition}
Every sentence of the form $\Esc{r}{u_1,\dots,u_M} \big( \phi_1 \loc r(u_1) \land \dots \land \phi_M \loc r(u_M) \land \Anotball x R {\tup u} \; \psi\loc{r'}(x) \big)$ is equivalent to a local sentence.
\end{proposition}

\begin{proof}
Let $\Phi$ be a sentence of the indicated form. We claim that $\Phi \meq \biglor_P \Psi^P$, where the disjunction ranges over all partitions $P = \{I_1,\dots,I_m\} \in \Part(M)$ with representatives $i_k = \min {I_k}$.
We define $\Psi^P \coloneqq (\Psi_1^P \land \Psi_2 \land \Psi_3^P)$ with the following sentences.
\begin{align*}
    \Psi_1^P &\coloneqq \Esc {2 \cdot 5^{M-m}(r+R)}{u_{i_1},\dots,u_{i_m}} \\
    &\phantom{\coloneqq\quad}
        \Eball {(u_i)_{i \in I_1 \setminus \{i_1\}}} {5^{M-m}(r+R)-(r+R)} {u_{i_1}} \\
    &\phantom{\coloneqq\quad\quad}
        \dots \\
    &\phantom{\coloneqq\quad\quad\quad}
        \Eball {(u_i)_{i \in I_m \setminus \{i_m\}}} {5^{M-m}(r+R)-(r+R)} {u_{i_m}} \\
    &\phantom{\coloneqq\quad\quad\quad\quad}
    \Big(
        \bigland_{i<j} d(u_i,u_j) > 2r \;\land\;
        \bigland_{i=1}^M \phi_i \loc r(u_i) \\
    &\phantom{\coloneqq\quad\quad\quad\quad\quad}
        \;\land\; \Aball{x}{2 \cdot 5^{M-m}(r+R)}{\tu} \big( d(x,\tup u) \le R \,\lor\, \psi \loc{r'}(x) \big)
    \Big).
\intertext{
The first sentence $\Psi_1^P$ first quantifies a tuple $\tu$ in configuration $(P,r+R)$ and ensures that it is $r$-scattered.
The remaining subformula is similar to $\Psi$, except that only some $x \notin \ball{R}{\tu}$ are considered (those in $\ball{2 \cdot 5^{M-m}(r+R)}{\tu}$).
}
    \Psi_2 &\coloneqq \Asc{R}{x_1,\dots,x_{M+1}} \biglor_{i=1}^{M+1} \psi\loc{r'}(x_i).
\intertext{We will show that $\Psi_2$ checks $\psi\loc{r'}(x)$ for most of the $x \notin \ball{R}{\tu}$. Finally,}
    \Psi_3^P &\coloneqq \A v_1 \dots \A v_M \;\biglor_{l=0}^m \;\; \biglor_{J=\{j_1,\dots,j_l\} \subseteq \{1,\dots,m\}} 
        \Eball {(u_{i_j})_{j \in J}\,}{3R + 5^{M-m}(r+R)}{\tv} \\
    &\phantom{\coloneqq\quad}    
        \Eball {(u_i)_{i \in I_{j_1} \setminus \{i_{j_1}\}}} {5^{M-m}(r+R)-(r+R)}{u_{i_1}} \\
    &\phantom{\coloneqq\quad\quad}    
        \dots \\
    &\phantom{\coloneqq\quad\quad\quad}    
        \Eball {(u_i)_{i \in I_{j_l} \setminus \{i_{j_l}\}}} {5^{M-m}(r+R)-(r+R)}{u_{i_l}} \\
    &\phantom{\coloneqq\quad\quad\quad\quad} 
        \Big( \big(\bigland_{j \in J} \bigland_{i \in I_j} \phi_i \loc r(u_i)\big) \; \land \; \bigland_{i < i'} d(u_i,u_{i'}) > 2r \\
    &\phantom{\coloneqq\quad\quad\quad\quad\quad\quad}
        \;\land\; \Aball{x}{2R}{\tv} \big( d(x,\tup u) \le R \  \lor \psi\loc{r'}(x) \big) \Big).
\end{align*}
The last sentence is the most involved one.
For every choice of $\tv$, there is a number $l$ such that the (partial) $r$-scattered tuple $\tu$ mimics the $l$ classes with representatives $i_{j_1},\dots,i_{j_l}$.
That is, there are elements $u_{i_{j_1}}, \dots, u_{i_{j_l}}$, and each $u_{i_{j_k}}$ has $|I_{j_k}|$ elements close to it.
However, note that the elements $u_{i_{j_1}}, \dots, u_{i_{j_l}}$ are not necessarily far from each other, as the configuration $(P,r+R)$ would require.
We further remark that, unless $l = m$, the tuple $\tu$ is smaller than $M$ and hence the notation $u_i$ is only defined for some $1 \le i \le M$; the idea is that $\tu$ shares indices with an $M$-tuple in configuration $(P,r+R)$.
We kindly ask the reader to bear with this slightly unusual indexing, as it will simplify the notation in the remaining proof.
$\Psi_3^P$ then checks $\phi \loc r(u_i)$ for all defined entries of $\tu$, and $\psi \loc {r'}(x)$ for all $x \notin \ball{R}{\tu}$ with $x \in \ball{2R}{\tv}$.

\medskip
We first observe that $\Psi_1^P$, $\Psi_2$ and $\Psi_3^P$ can all be written as local sentences (notice that $\Psi_2$ may not be a basic local sentence yet, as we may have $R \neq r'$).
Indeed, all sentences use only one type of quantifiers followed by a local formula (of sufficiently large radius).
We can thus apply \cref{thmGaifmanQuantifiers} or its dual version (scattered quantifiers can first be replaced by regular quantifiers and local distance formulae).

\medskip
It remains to prove equivalence.
In the following, fix a $K$-interpretation $\pi$.
We first prove that there is a partition $P$ with $\pi \ext {\Psi^P} \ge \pi \ext \Phi$.
Let $\td$ be an $M$-tuple at which the maximum in the evaluation of $\pi \ext \Phi$ is reached.
Then $\td$ is $r$-scattered (if no $r$-scattered tuple exists, then $\pi \ext \Phi = 0$ and the inequality is trivial).
By the \ClusteringLemma{}, there is a partition $P = \{I_1,\dots,I_m\} \in \Part(M)$ such that $\td$ is in configuration $(P, r+R)$.

Then $\pi \ext {\Psi_1^P} \ge \pi \ext \Phi$ by using $\td$ as witness for $\tu$ in $\Psi_1^P$.
Indeed, the main difference is that the evaluation of $\Phi$ includes the minimum of $\pi \ext {\psi\loc{r'}(a)}$ over all $a \notin \ball{R}{\td}$, whereas the evaluation of $\Psi_1^P$ effectively only considers $a \notin \ball{R}{\td}$ with $a \in \ball{2 \cdot 5^{M-m}(r+R)}{\td}$.

To see $\pi \ext {\Psi_2} \ge \pi \ext \Phi$, let $\tb$ be an $(M+1)$-tuple at which the minimum is reached.
Since $\tb$ is $R$-scattered, there must be at least one entry $b_i$ with $b_i \notin \ball R {\td}$.
Hence $\pi \ext {\Psi_2} \ge \pi \ext {\psi\loc{r'}(b_i)} \ge \pi \ext {\Anotball x R {\td} \; \psi\loc{r'}(x)}$.

For $\Psi_3^P$, let $\tv$ be an $M$-tuple at which the minimum is reached.
Recall that $\Phi$ reaches its maximum at $\td$.
Let $I = \{ i \mid$ there is $k$ with $i \in I_{i_k}$ and $d_{i_k} \in \ball{3R + 5^{M-m}(r+R)}{\tv} \}$ be the set of indices whose representative element is close to $\tv$.
Define $\td'$ as the subtuple of $\td$ consisting of those entries $d_i' = d_i$ with $i \in I$ (the other entries are undefined) and use $\td'$ as witness for $\tu$ in $\Psi_3^P$.
Clearly $\pi \ext {d(d'_i,d'_{i'}) > 2r} = 1$ since $\td$ is $r$-scattered,
and further $\pi \ext {\bigland_{j \in J} \bigland_{i \in I_j} \phi_i \loc r(d'_i)} \ge \pi \ext {\bigland_{i=1}^M \phi_i \loc r(d_i)}$.
For the remaining formula involving $\psi \loc {r'}$, the argument is similar to $\Psi_1^P$.
Here we effectively take the minimum of $\pi \ext {\psi\loc{r'}(a)}$ for those $a \in \ball {2R} {\tup v}$ that are $R$-far from $\td'$.
We claim that all such elements $a$ are also $R$-far from the full tuple $\td$.
To see this, consider an entry $d_i$ with $i \notin I$ and let $k$ be such that $i \in I_k$.
We have the following distances:
\begin{align*}
    d(a,\tv) &\le 2R, \\
    d(\tv,d_{i_k}) &> 3R + 5^{M-m}(r+R), &&\text{by construction of $\td'$}, \\
    d(d_{i_k},d_i) &\le 5^{M-m}(r+R)-(r+R), &&\text{since $\td$ is in configuration $(P,r+R)$}.
\end{align*}
Together, they imply $d(a,d_i) > R + (r+R) \ge R$, so $a$ is indeed $R$-far from $d_i$.
Hence in $\Psi_3$ we take the minimum of $\pi \ext {\psi\loc{r'}(a)}$ over (a subset of) elements $a \notin \ball {R} {\tup d}$.
We have shown $\pi \ext {\Psi_3^P} \ge \pi \ext \Phi$, and thus $\pi \ext {\Psi^P} \ge \pi \ext \Phi$.

\bigskip
Finally, we prove that $\pi \ext {\Psi^P} \le \pi \ext \Phi$, for every partition $P$.
This implies equivalence.
To prove the inequality, we fix a partition $P$ and construct an $r$-scattered $M$-tuple $\td$ with
\[
    \pi \ext {\Psi^P} \le \pi \ext {\bigland_{i=1}^M \phi_i \loc r(d_i) \,\land\, \Anotball x R {\td} \, \psi\loc{r'}(x)} \le \pi \ext \Phi. \tag{$*$}
\]
We construct $\td$ by taking each entry $d_i$ either from an $M$-tuple $\ta$, or from a partial $M$-tuple $\tc$.
This is done in such a way that all entries from one partition class $I_k$ (i.e., all $d_i$ with $i \in I_k$) are taken from the same tuple.

Let $\ta$ be an $M$-tuple at which the maximum in $\Psi_1^P$ is reached.
Then $\ta$ is $r$-scattered and in configuration $(P,r+R)$ (if no such tuple exists, then $\pi \ext {\Psi_1^P} = 0$ and the inequality is trivial).
Let further $\tb$ be a maximal $R$-scattered tuple with $\pi\ext{\Psi_2} > \pi\ext{\psi\loc{r'}(b_i)}$ for all $i$.
The size of $\tb$ is at most $M$ (by construction of $\Psi_2$) and $\pi\ext{\Psi_2} \le \pi\ext{\psi\loc{r'}(b')}$ for all $b' \notin \ball{2R}{\tb}$ (by maximality).
If necessary, we extend $\tb$ to an $M$-tuple by adding copies of $b_1$ (the resulting tuple is no longer $R$-scattered; this will not matter).
Finally, let $\tc$ be the partial tuple (together with $l$ and $J = \{j_1,\dots,j_l\}$) at which the maximum in $\Psi_3^P$ is reached once we instantiate $\tv$ by $\tb$, so
\begin{align*}
    \pi \ext {\Psi_3^P} \le
    \pi \Ext {\bigland_{\substack{j \in J\\i \in I_j}} \phi_i \loc r(c_i) \;\land\; \bigland_{i < i'} d(c_i,c_{i'}) > 2r \;\land\; \Aball{x}{2R}{\tb} \big( d(x,\tc) \le R  \,\lor\, \psi\loc{r'}(x)\big)}.
\end{align*}
In particular, $\tc$ is $r$-scattered (unless $\pi \ext \Psi_3^P = 0$).
Note that the entry $c_i$ is defined only for $i \in I_k$ with $k \in J$.
We remark that $l=0$ is possible, so $\tc$ may be empty (then $\pi \ext {\Psi_3^P} \le \pi \ext {\Aball{x}{2R}{\tb} \, \psi\loc{r'}(x)}$); this will not affect our arguments.

\medskip
We now define the desired tuple $\td$.
Let $\bar J$ denote the complement $\{1,\dots,m\} \setminus J$.
For each $K \subseteq \{1,\dots,m\}$ with $\bar J \subseteq K$, we define an $M$-tuple $\td^K$ by
\[
    d^K_i = \begin{cases*}
        a_i, &if $i \in I_k$ with $k \in K$, \\
        c_i, &if $i \in I_k$ with $k \notin K$.
    \end{cases*}
\]
That is, $K$ specifies for which partition classes we take the elements from $\ta$, and for the other classes we use the entries from $\tc$ (which are defined, as $\bar J \subseteq K$).
It remains to find the right choice for $K$ so that $\td^K$ is $r$-scattered and satisfies $(*)$.
We use the following algorithm that starts by taking all possible entries from $\tc$, and then switches (classes of) entries to $\ta$ until the resulting tuple is $r$-scattered.

\begin{algorithm}[H]
    initialise $K = \bar J$\;
    \While{there are $k \in K$, $k' \notin K$ and $i \in I_k$, $i' \in I_{k'}$ with $d(a_i,c_{i'}) \le 2r$}{
        $K \gets K \cup \{k'\}$\;
    }
\end{algorithm}

\smallskip
This algorithm clearly terminates, as the loop condition is violated for $K = \{1,\dots,m\}$.
Let $\td = \td^K$ be the resulting tuple.
Then $\td$ is $r$-scattered: the tuples $\ta$ and $\tc$ are both $r$-scattered, and whenever $d_i = a_i$ and $d_{i'} = c_i$ with $d(a_i,c_i) \le 2r$, the loop condition applies to $i$, $i'$.
We conclude the proof by showing the following two claims.

\begin{claim*}
$\pi \ext {\Psi^P} \le \pi \ext {\phi_i \loc r (d_i)}$ for all $i$
\end{claim*}
\begin{claimproof}
By construction of $\ta$, we have $\pi \ext {\Psi_1^P} \le \pi \ext {\phi_i \loc r (a_i)}$ for all $1 \le i \le M$.
By construction of $\tb$ and $\tc$, we have $\pi \ext {\Psi_3^P} \le \pi \ext {\phi_i \loc r (c_i)}$ for all $i$ where $c_i$ is defined.
The claim follows.
\end{claimproof}

\begin{claim*}
$\pi \ext {\Psi^P} \le \pi \ext {\Anotball x R {\td^K} \, \psi \loc{r'}(x)}$ is an invariant of the algorithm.
\end{claim*}
\begin{claimproof}
For $x \notin \ballpi {2R} {\tb}$, we have $\pi \ext {\Psi_2} \le \pi \ext {\psi \loc{r'}(x)}$ (by choice of $\tb$).
For $x \in \ballpi {2R} {\tb} \setminus \ballpi R \tc$, we have $\pi \ext {\Psi_3^P} \le \pi \ext {\psi \loc{r'}(x)}$.
Since $\td^K$ initially contains all (defined) entries of $\tc$, we have $\pi \ext {\Psi_2 \land \Psi_3^P} \le \pi \ext {\Anotball x R {\td^K} \, \psi \loc{r'}(x)}$ at the start of the algorithm.

Now suppose the invariant is true for $K$ and we perform an update to $K \cup \{k'\}$, where we have $k \in K$, $k' \notin K$ and $i \in I_k$, $i' \in I_{k'}$ with $d(a_i,c_{i'}) \le 2r$.
Let $x \notin \ballpi R {\td^{K \cup \{k'\}}}$.
If $x \notin \ballpi R {\td^{K}}$, then the invariant for $K$ applies, so suppose $x \in \ballpi R {\td^{K}}$.
We prove that $\pi \ext {\Psi_1^P} \le \pi \ext {\psi\loc{r'}(x)}$, which implies the claim.
We need to prove two properties:
\begin{itemize}
\smallskip
\item $x \in \ballpi {2 \cdot 5^{M-m}(r+R)}{\ta}$:\\
Since $x \in \ballpi R {\td^{K}} \setminus \ballpi R {\td^{K \cup \{k'\}}}$, there is $j \in I_{k'}$ such that $x \in \ballpi R {c_j}$.
We obtain the following distances:
\begin{align*}
    d(x,c_j) &\le R, \\
    d(c_j,c_{i_{k'}}) &\le 5^{M-m}(r+R)-(r+R), &&\text{since $j \in I_{k'}$}, \\
    d(c_{i_k'},c_{i'}) &\le 5^{M-m}(r+R)-(r+R), &&\text{since $i' \in I_{k'}$ (loop condition)}, \\
    d(c_{i'}, a_i) &\le 2r, &&\text{by the loop condition}.
\end{align*}
Together, they imply $d(x,a_i) \le 2 \cdot 5^{M-m}(r+R) - R$, thus $x \in \ballpi {2 \cdot 5^{M-m}(r+R)}{\ta}$.

\medskip
\item $x \notin \ballpi R \ta$:\\
Since $x \notin \ballpi R {\td^{K \cup \{k'\}}}$ and $k \in K$, we have $c \notin \ballpi R {a_j}$ for all $j \in I_k$.
Now consider $j \in I_{k''}$ for any $k'' \neq k$.
Then:
\begin{align*}
    d(x,a_i) &\le 2 \cdot 5^{M-m}(r+R) - R, &&\text{see above} \\
    d(a_i,a_{i_k}) &\le 5^{M-m}(r+R)-(r+R), &&\text{since $i \in I_{k}$}, \\
    d(a_{i_k},a_{i_{k''}}) &> 4 \cdot 5^{M-m}(r+R), &&\text{scattered quantification in $\Psi_1^P$}, \\
    d(a_{i_{k''}}, a_j) &\le 5^{M-m}(r+R)-(r+R), &&\text{since $j \in I_{k''}$}.
\end{align*}
Together, they imply $d(x,a_j) > R + 2(r+R) \ge R$. \claimqedhere
\end{itemize}
\end{claimproof}

\noindent
Combining the two claims implies $\pi \ext {\Psi^P} \le \pi \ext \Phi$, which closes the equivalence proof.
\end{proof}

This ends the proof of step \enumref{3}, and thus the proof of \cref{thmGaifman}.

\section{Strengthening Gaifman's Theorem}
\label{secStrengthening}

In this section, we rephrase our main result in terms of Boolean semantics, which leads to a novel strengthening of Gaifman's classical theorem.
Interestingly, \cref{thmGaifman} can be regained from the Boolean result by algebraic techniques, and even lifted to the class of lattice semirings (denoted $\Lattice$).
These insights suggest that the merit of our proof, and the reason why it is more complicated than Gaifman's original proof, is not primarily the more fine-grained notion of equivalence, but rather the construction of a Gaifman normal form without the use of negation.
Indeed, a careful examination of our proof in \cref{secGaifmanProof} reveals that our constructions only use literals that were already part of the original sentence, and hence do not add negations.

Our proof applies in particular to the Boolean semiring, and we can thus formulate this observation for standard Boolean semantics.
The only difference between semiring semantics in the Boolean semiring and Boolean semantics is our addition of ball quantifiers, but these can be expressed by distance formulae.
Moreover, we always assume negation normal form and thus permit the duals of basic local sentences (quantifying over \emph{all} scattered tuples).
In Boolean semantics, we can instead express these by negations of basic local sentences.
Notice that this adds negations to the formula, but the number of negations added in front of each literal is even.
Following the common definition that an occurrence of a relation is \emph{positive} if it is behind an \emph{even} number of negations, and \emph{negative} if it is behind an \emph{odd} number of negations, we can formulate the following strengthening of Gaifman's classical result.

\begin{corollary}[Gaifman normal form without negation]
\label{thmGaifmanNoNegation}
Let $\tau$ be a finite relational signature.
In Boolean semantics, every $\FO(\tau)$-sentence $\psi$ has an equivalent local sentence $\psi'$ such that
every relation symbol occurring only positively (only negatively) in $\psi$ also occurs only positively (only negatively) in $\psi'$, not counting occurrences within distance formulae.
\end{corollary}

Notice that \cref{thmGaifmanNoNegation} in particular says that if a relation symbol does not occur at all in $\psi$ (i.e., neither positively nor negatively), then it also does not occur in $\psi'$.
We make use of this property to prove \cref{thmGaifman} below.
We further remark that \cref{thmGaifmanNoNegation} only applies to sentences (see \cref{exGaifmanNoNegationCounter} below).

We believe that this result may be of independent interest. 
A similar adaptation of Gaifman's theorem has been considered in \cite{GroheWoe04}, namely that \emph{existential} sentences are equivalent to \emph{positive} Boolean combinations of \emph{existential} basic local sentences.
We prove a similar result (see \cref{remarkExistentialGaifman}), as we also construct a positive Boolean combination of existential basic local sentences (but we permit distance formulae $d(x,y) > 2r$, while \cite{GroheWoe04} does not).
Moreover, the approximation schemes of \cite{DawarGroKreSch06} are based on a version of Gaifman's theorem for sentences positive in a single unary relation (i.e., no negations are added in front of this relation).
Their proof uses a version of Ehrenfeucht-Fra\"{\i}ss{\'e} games, which is quite different from our syntactical approach.
Since unary relations do not occur in distance formulae, \cref{thmGaifmanNoNegation} subsumes their result.
Interestingly, \cite{GroheWoe04,DawarGroKreSch06} both share our observation that the proof of the respective version of Gaifman's theorem is surprisingly difficult.

\medskip
It turns out that we can prove \cref{thmGaifman} just from \cref{thmGaifmanNoNegation}, so we could actually rephrase the proof in \cref{secGaifmanProof} in terms of Boolean semantics.
But since the main difficulty is about negation rather than semirings, this would not lead to a significant simplification.
The proof uses the method of \emph{separating homomorphisms} \cite{GraedelMrk21}.
After some preparations, these homomorphisms allow us to lift Boolean equivalences $\keq[\Bool]$ to minmax-equivalences $\meq$ (thus proving \cref{thmGaifman}) and even to equivalences on lattice semirings $\keq[\Lattice]$.

\begin{definition}[Separating homomorphisms \cite{GraedelMrk21}]
Let $K,L$ be two semirings.
A set $S$ of homomorphisms $h \from K \to L$ is called \emph{separating} if for all $a,b \in K$ with $a \neq b$,
there is a homomorphism $h \in S$ such that $h(a) \neq h(b)$.
\end{definition}

Here we consider homomorphisms into the Boolean semiring $L = \Bool$.
For any min-max semiring $K$ and every non-zero element $b\in K$, we define the homomorphism
\[
    h_b \from K \to \Bool, \quad x \mapsto \begin{cases*}
        1, &if $x \ge b$, \\
        0, &if $x < b$.
    \end{cases*}
\]
Then, for any pair $a<b$ in K, we have that $h_b(b)=1$ and $h_b(a)=0$, so these homomorphisms form a separating set.

We also need the following technical observation.
Here we use the notation $\phi \sig \tau$ introduced in the \AbstractionLemma{} to restrict the signature of ball quantifiers.

\begin{lemma}
\label{thmEquivalenceSignature}
Let $\tau' = \tau \dcup \{R\}$ and let $\phi,\psi$ be $\tau'$-formulae in which $R$ does not occur (but quantifiers $\Qballtau {\tau'} y r x$ implicitly depend on $R$).
Then $\phi \meq \psi$ implies $\phi \sig \tau \meq \psi \sig \tau$.
\end{lemma}

\begin{proof}
Let $\pi$ be a model-defining $K$-interpretation of signature $\tau$.
Extend $\pi$ to a model-defining $K$-interpretation $\pi'$ of signature $\tau'$ by adding an empty relation $R$, i.e., $\pi'(R\ta) = 0$ for all $\ta$.
Then $G(\pi) = G(\pi')$, so quantifiers $\Qballtau {\tau} y r a$ in $\pi$ and $\Qballtau {\tau'} y r a$ in $\pi'$ range over the same $r$-neighbourhoods $\ballpi r a = \ballpi[\pi'] r a$.
Hence $\pi \ext {\phi \sig \tau} = \pi' \ext \phi = \pi' \ext \psi = \pi \ext {\psi\sig\tau}$, where the equality in $\pi'$ is due to $\phi \meq \psi$.
\end{proof}

We now apply separating homomorphisms to lift $\keq[\Bool]$ to $\meq$.
The idea is that a falsifying interpretation $\pi$ for $\meq$ induces a falsifying interpretation $h \circ \pi$ for $\keq[\Bool]$.
However, we need some preparation to make sure that the interpretation $h \circ \pi$ remains model-defining and that the Gaifman graph is preserved.
We achieve this by extending $\pi$ with a relation $G$ encoding its Gaifman graph, and making sure that $G$ is preserved by the homomorphism $h$.

\begin{proof}[Proof of \cref{thmGaifman} from \cref{thmGaifmanNoNegation}]
Let $\psi \in \FO(\tau)$ be a sentence in negation normal form,
and let $\tau' = \tau \cup \{G\}$ for a fresh binary relation symbol $G$.
For the rest of the proof, we view $\psi$ as a $\tau'$-sentence.

If there is a relation symbol $R$ that occurs both positively and negatively in $\psi$, we apply the \AbstractionLemma{} to replace all positive occurrences $R\tx$ by $R_1 \tx$ and all negative occurrences $\neg R\tx$ by $\neg R_0 \tx$, for fresh relation symbols $R_0,R_1$.
Let $\psi'$ be the resulting $\tau''$-sentence (with $\tau' \subseteq \tau''$) in which every relation symbol occurs either only positively, or only negatively (and $G$ does not occur at all).

By \cref{thmGaifmanNoNegation}, there is a Gaifman normal form $\phi_\Bool'$ (of signature $\tau''$) of $\psi'$ such that also for $\phi_\Bool'$, every relation symbol occurs only positively, or only negatively, and $G$ does not occur (outside of distance formulae).
We bring $\phi_\Bool'$ to negation normal form and use ball quantifiers instead of distance formulae.
The resulting formula $\phi'$ is a Gaifman normal form according to our definition.
We have $\psi' \keq[\Bool] \phi'$ and we claim that also $\psi' \meq \phi'$.
Then, the \AbstractionLemma{} implies $\psi \meq \phi$, where $\phi$ is a Gaifman normal form (of signature $\tau'$) that results from $\phi'$ by replacing all $R_0,R_1$ back.
By \cref{thmEquivalenceSignature}, $\psi \meq \phi$ implies $\psi \meq \phi \sig \tau$, so $\phi \sig \tau$ is a Gaifman normal form for the $\tau$-sentence $\psi$, which closes the proof.

It remains to prove $\psi' \meq \phi'$ (over the signature $\tau''$).
Towards a contradiction, assume that $\psi' \not\meq \phi'$.
Then there is a min-max semiring $K$ and a model-defining $K$-interpretation $\pi$ over universe $A$ and signature $\tau''$ such that $\pi \ext {\psi'} = s \neq t = \pi \ext {\phi'}$ for some $s,t \in K$.
We define $\pi'$ by modifying $\pi$ as follows: for all $a,b \in A$, we set $\pi'(Gab) = 1$ if $a,b$ are adjacent in $G(\pi)$, and $\pi'(Gab)=0$ otherwise.
Then $G(\pi') = G(\pi)$ and since $G$ does not occur in $\psi',\phi'$, we still have $\pi' \ext {\psi'} = s \neq t = \pi' \ext {\phi'}$.
Let $h$ be a separating homomorphism for $s,t$.
We define a $\Bool$-interpretation $\pi_\Bool$ over $A$ as follows.
For every relation symbol $R \in \tau''$,
\begin{itemize}
\item if $R$ occurs only positively in $\psi',\phi'$, we set $\pi_\Bool(R\ta) = h(\pi'(R\ta))$ for all $\ta \subseteq A$;
\item if $R$ occurs only negatively in $\psi',\phi'$, we set $\pi_\Bool(\neg R\ta) = h(\pi'(\neg R\ta))$ for all $\ta \subseteq A$.
\end{itemize}

We define the unspecified values in the unique way so that $\pi_\Bool$ is model-defining.
Notice that these values correspond to literals that do not occur in $\psi',\phi'$, so they do not affect the evaluation of $\psi',\phi'$ in $\pi_\Bool$.
Since $h(0)=\bot$, false literals remain false, so $G(\pi_\Bool) \subseteq G(\pi)$.
And since $h(1)=\top$, the interpretation of $G$ is preserved and we have equality: $G(\pi_\Bool) = G(\pi)$.
Then $\pi_\Bool \ext {\psi'} = h(\pi' \ext {\psi'}) \neq h(\pi' \ext {\phi'}) = \pi_\Bool \ext {\phi'}$ by a straightforward induction (cf.\ \cite[Fundamental Property]{GraedelTan17}), where we use $G(\pi) = G(\pi_\Bool)$ to show that the semantics of ball quantifiers is preserved.
But this contradicts $\psi' \keq[\Bool] \phi'$.
\end{proof}

We remark that the lifting argument implies that for many sentences (to be precise, those where no relation occurs both positively and negatively), the Gaifman normal form in min-max and lattice semirings coincides with the one for Boolean semantics in \cref{thmGaifmanNoNegation} (but not necessarily with Gaifman's original construction).

\begin{example}
\label{exGaifmanSimple}
Consider the sentence $\psi = \E x \A y \, Exy$ which asserts that every node is adjacent to a central node $x$.
In particular, the diameter of the Gaifman graph must be at most $2$, so it suffices to locally quantify with radius $2$.
We thus obtain the following Gaifman normal form in Boolean semantics:
\[
    \psi \equiv \neg \E x_1 \E x_2 (d(x_1,x_2) > 2 \,\land\,\text{true}) \;\land\; \E x \A y (d(x,y) \le 2 \to Exy).
\]

The same Gaifman normal form also works in all min-max semirings (up to negation normal form and ball quantifiers; $\text{false}$ can be expressed by the local formula $x \neq x$):
\[
    \psi \meq \Asc 1 {x_1, x_2} \,\text{false} \;\land\; \E x \, \Aball y 2 x \, Exy. \tag*{\lipicsEnd}
\]
\end{example}

\medskip
A further consequence is that the counterexample for formulae also applies to \cref{thmGaifmanNoNegation}.

\begin{Example}
\label{exGaifmanNoNegationCounter}
Recall the counterexample $\psi(x) = \E y (Uy \land y \neq x)$ of \cref{secCounterexampleFormula}.
In Boolean semantics, a Gaifman normal form is given by
\[
    \psi \equiv \E x_1 \E x_2 (x_1 \neq x_2 \land Ux_1 \land U x_2) \;\lor\; (\neg Ux \,\land\, \E y \, U y).
\]
Notice that this Gaifman normal form \emph{adds negations}: the relation $U$ occurs only positively in $\psi$, but here occurs negatively in $\neg U x$.
This is not a coincidence, but must be the case for every Gaifman normal form of $\psi(x)$.
Otherwise, the reasoning in the above proof would imply that the Gaifman normal form also works in min-max semirings, contradicting \cref{secCounterexampleFormula}.
This shows that \cref{thmGaifmanNoNegation} does, in general, not hold for formulae.
\end{Example}

\medskip
Finally, we remark that the proof of \cref{thmGaifman} via separating homomorphisms is not specific to min-max semirings, but applies to any class of semirings for which separating homomorphisms into the Boolean semiring exist.
It can be shown that such separating homomorphisms exist for every lattice semiring (see \cite{Brinke23}).
We can thus generalise \cref{thmGaifman} to lattice semirings.

\begin{restatable}{corollary}{restateGaifmanLattice}
Let $\tau$ be a finite relational signature.
Every $\FO(\tau)$-sentence $\psi$ is lattice-equivalent ($\keq[\Lattice]$) to a local sentence.
\end{restatable}

\section{Conclusion}

Semiring semantics is a refinement of classical Boolean semantics, which provides more detailed information
about a logical statement than just its truth or falsity. This leads to a finer distinction between 
formulae: statements that are equivalent in the Boolean sense may have different valuations in
semiring interpretations, depending on the underlying semiring. It is an interesting and non-trivial question,
which logical equivalences and, more generally, which model-theoretic methods, can be carried over from
classical semantics to semiring semantics, and how this depends on the algebraic properties of the underlying 
semiring.

Here we have studied this question for locality properties of first-order logic, in particular for Hanf's locality 
theorem and for Gaifman normal forms. Our setting assumes semiring interpretations which are model-defining and track only positive information, since these
are the conditions that provide well-defined and meaningful locality notions. However, from the outset, it has been clear that one cannot
expect to transfer all locality properties of
first-order logic to semiring semantics in arbitrary commutative semirings. 
Indeed, semiring semantics evaluates existential and universal quantifiers  
by sums and products over all elements of the universe,
which gives an inherent source of non-locality if these operations are not idempotent. 

Most positive locality results thus require that the underlying semirings
are fully idempotent. Under this assumption, one can adapt the classical proof
of Hanf's locality theorem to the semiring setting, relying on a back-and-forth argument that
itself requires fully idempotent semirings.
The question whether there exist Gaifman normal forms in semiring semantics turned out to be
more subtle. Indeed, for formulae with free variables Gaifman normal forms need not
exist once one goes beyond the Boolean semiring. Also for sentences, one can find
examples that do not admit Gaifman normal forms in semirings that are not fully idempotent.
We have presented such an example for the tropical semiring.

Our main result, however, is a positive one and establishes the existence of Gaifman normal forms over the class of 
all min-max and lattice semirings.
Intuitively, it relies on the property that in min-max semirings, the value of a quantified statement $\E x \, \phi(x)$ or $\A x \, \phi(x)$
coincides with a value of $\phi(a)$, for some witness $a$.
This needs, for instance, not be the case in lattice semirings, and hence the generalisation to lattice semirings uses a different approach based on separating homomorphisms.
It is still an open question whether, in analogy to Hanf's theorem, Gaifman normal forms exist over all fully idempotent semirings.
The proof of our main result, which is based on quantifier elimination arguments, turned out to be surprisingly difficult;
we identified the lack of a classical negation operator as the main reason for its complexity.
An interesting consequence of this restriction is a stronger version of Gaifman's classical theorem in Boolean semantics:
every sentence has a Gaifman normal form which, informally speaking, does not add negations.

For applications such as provenance analysis, min-max semirings are relevant, for instance, for studying access levels and security issues.
A much larger interesting class of semirings with wider applications
are the absorptive ones, including the tropical semiring, in which addition is idempotent, but multiplication in general is not. We have seen that
Gaifman normal forms for such semirings need not exist for all sentences. The question arises whether one can establish
weaker locality properties for absorptive semirings, applicable perhaps to just a relevant fragment of first-order logic.

%%%%%%%%%%%%%%%%%%%%%%%%%%%%%%%%
%%  REFERENCES
%%%%%%%%%%%%%%%%%%%%%%%%%%%%%%%%

\bibliography{locality}

%%%%%%%%%%%%%%%%%%%%%%%%%%%%%%%%
%%  APPENDIX
%%%%%%%%%%%%%%%%%%%%%%%%%%%%%%%%

\appendix
\clearpage

% special numbering for appendix (needs thm-restate option)
\counterwithin{theorem}{section}
\counterwithin{lemma}{section}
\counterwithin{proposition}{section}
\counterwithin{corollary}{section}

\section{Counterexample against Distance Formula}
\label{appendixDistance}

In this appendix, we justify the addition of ball quantifiers $\Qball y r x$ by showing that if we define Gaifman normal forms
via the same distance formulae $d(x,y) \le r$ as in Boolean semantics, then \cref{thmGaifman} does not hold.
For simplicity, we only consider graphs.
In Boolean semantics, we can then define the distance formula inductively:
\begin{align*}
    d(x,y) \le 0 &\;\coloneqq\; x=y, \\
    d(x,y) \le r+1 &\;\coloneqq\; \E z ((Exz \lor Ezx) \land d(z,y) \le r).
\end{align*}
We use the same formula also in semiring semantics.
For example, for $r=2$ we have (after slight simplifications, which are also sound in min-max semirings):
\[
    d(x,y) \le 2 \;\equiv\; x=y \,\lor\, (Exy \lor Eyx) \,\lor\, \E z ((Exz \lor Ezx) \land (Ezy \lor Eyz)).
\]
We can then express $d(x,y) > r$ as negation, where we use negation normal form in semiring semantics.
For example,
\[
    d(x,y) > 2 \;\equiv\; x \neq y \,\land\, (\neg Exy \land \neg Eyx) \,\land\, \A z ((\neg Exz \land \neg Ezx) \lor (\neg Ezy \land \neg Eyz)).
\]

We then use these distance formulae to define Gaifman normal forms.
In a local formula $\phi \loc r(x)$, all quantifiers $\E y \, \theta(x,y)$ are relativised to $\E y (d(x,y) \le r \land \theta(x,y))$,
and quantifiers $\A y \, \theta(x,y)$ become $\A y (d(x,y) > r \lor \theta(x,y))$.
Basic local sentences are defined as in \cref{defLocalSentence}, where scattered quantifiers are expressed using distance formulae:
\begin{align*}
    \Esc r {y_1,\dots y_m} \, \theta(\ty) &\coloneqq \E y_1 \dots \E y_m \big( \textstyle\bigland_{i < j} d(y_i,y_j) > 2r \land \theta(\ty) \big), \\
    \Asc r {y_1,\dots y_m} \, \theta(\ty) &\coloneqq \A y_1 \dots \A y_m \big( \textstyle\biglor_{i < j} d(y_i,y_j) \le 2r \lor \theta(\ty) \big).
\end{align*}
%(see the text preceding \cref{defLocalSentence}).
%We make the following claim:

\smallskip
\begin{proposition}
\label{thmDistanceCounterexample}
The sentence $\beta \coloneqq \E x \E y (Ux \land \neg Exy)$ does not have a Gaifman normal form with distance formulae (as defined above) over min-max semirings.
\end{proposition}

\begin{proof}
\resetlocalclaim
We construct families of almost identical $K$-interpretations $\pi_n$, $\pi'_n$ over the same increasingly large universe $A_n = \{v,w,v_1,\dots,v_n\}$, and then show that every potential Gaifman normal form behaves differently than $\beta$ on $\pi_n$ or $\pi'_n$ for sufficiently large $n$.

Let $K$ be a min-max semiring with at least 3 elements, and let $s,t \in K$ with $0 < s < t \le 1$.
For every $n \ge 1$, we define $\pi_n$ as the complete directed graph on $A_n$ with semiring values $\pi_n(Ux) = t$ and $\pi_n(Exy) = s$ for all $x,y \in A_n$ (and $\pi_n(\neg Ux) = \pi_n(\neg Exy) = 0$ for all $x,y$ so that $\pi_n$ is model-defining).
The second interpretation $\pi'_n$ results from $\pi_n$ by removing the edge $v \to w$, i.e., $\pi'_n(Evw) = 0$ and $\pi'_n(\neg Evw) = 1$ (see \cref{figDistanceCounterexample}).
Then $G(\pi_n) = G(\pi_n')$ and for the distance formulae, we have for $r \ge 1$ and all $a,b \in A$, all $n$ and $\pi \in \{\pi_n,\pi_n'\}$:
\begin{align*}
    \pi \ext {d(a,b) \le r} = \begin{cases*}
        1, &if $a=b$, \\
        s, &if $a \neq b$,
    \end{cases*}
    \quad \text{ and } \quad
    \pi \ext {d(a,b) > r} = 0.
\end{align*}

For the counterexample $\beta$, we have $\pi_n \ext \beta = 0$ and $\pi'_n \ext \beta = t$.
Suppose for contradiction that $\beta$ is equivalent to a positive Boolean combination $\gamma$ of basic local sentences (with distance formulae).
Assuming DNF, it has the form $\gamma = \gamma_1 \lor \dots \lor \gamma_k$ where each $\gamma_i$ is a conjunction of basic local sentences.
For the contradiction, we shall prove that every such sentence with $\pi_n \ext \gamma = 0$ also satisfies $\pi_n' \ext \gamma \le s < t$ for sufficiently large $n$, and can thus not be equivalent to $\beta$.

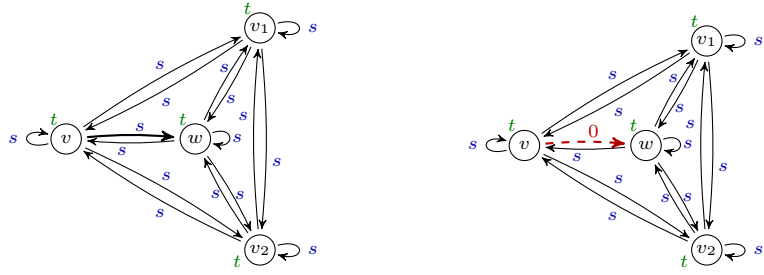
\begin{figure}[t]
\colorlet{csemi}{blue!70!black}
\colorlet{colu}{green!50!black}
\colorlet{cmark}{red!70!black}
\tikzset{vtx/.style={draw=black,circle,inner sep=0pt,minimum size=4mm,font=\scriptsize}}
\centering
\begin{tikzpicture}[scale=0.85]
    \node [vtx] (b) {$w$};
    \node [vtx] (a) at (180:2cm) {$v$};
    \node [vtx] (a1) at (60:2cm) {$v_1$};
    \node [vtx] (a2) at (-60:2cm) {$v_2$};
    \path [->,>=stealth',shorten <=2pt,shorten >=2pt,every node/.style={font=\scriptsize}]
        (a) edge [bend left=5pt,thick] node [above,csemi,yshift=-2pt,pos=0.6] {$s$} (b)
        (b) edge [bend left=5pt] node [below,csemi,yshift=2pt,pos=0.6] {$s$} (a)
        (a) edge [bend left=5pt] node [above,csemi] {$s$} (a1)
        (a1) edge [bend left=5pt] node [below,csemi] {$s$} (a)
        (a) edge [bend left=5pt] node [above,csemi] {$s$} (a2)
        (a2) edge [bend left=5pt] node [below,csemi] {$s$} (a)
        (b) edge [bend left=5pt] node [above,csemi] {$s$} (a1)
        (a1) edge [bend left=5pt] node [below,csemi] {$s$} (b)
        (b) edge [bend left=5pt] node [pos=0.7,above,csemi] {$s$} (a2)
        (a2) edge [bend left=5pt] node [pos=0.7,below,csemi] {$s$} (b)
        (a1) edge [bend left=5pt] node [pos=0.63,right,csemi,inner sep=2pt] {$s$} (a2)
        (a2) edge [bend left=5pt] node [pos=0.63,left,csemi,inner sep=2pt] {$s$} (a1)
        ;
\path [->,>=stealth',shorten <=1pt,shorten >=1pt,every node/.style={font=\scriptsize}]
        (a) edge [loop left] node [left,csemi] {$s$} (a)
        (b) edge [in=-15,out=20,looseness=7,shorten <=1pt,shorten >=1pt] node [right,csemi,inner sep=1pt] {$s$} (b)
        (a1) edge [loop right] node [right,csemi] {$s$} (a1)
        (a2) edge [loop right] node [right,csemi] {$s$} (a2)
        ;
\node [font=\scriptsize,colu,anchor=south east,inner sep=0pt] at (a.120) {$t$};
\node [font=\scriptsize,colu,anchor=south east,inner sep=0pt] at (b.120) {$t$};
\node [font=\scriptsize,colu,anchor=south east,inner sep=0pt] at (a1.120) {$t$};
\node [font=\scriptsize,colu,anchor=north east,inner sep=1pt] at (a2.190) {$t$};
\end{tikzpicture}
\hspace{1.5cm}
\begin{tikzpicture}[scale=0.8]
    \node [vtx] (b) {$w$};
    \node [vtx] (a) at (180:2cm) {$v$};
    \node [vtx] (a1) at (60:2cm) {$v_1$};
    \node [vtx] (a2) at (-60:2cm) {$v_2$};
    \path [->,>=stealth',shorten <=2pt,shorten >=2pt,every node/.style={font=\scriptsize}]
        (a) edge [bend left=5pt,thick,cmark,dashed] node [above,cmark,yshift=-2pt,pos=0.6] {$0$} (b)
        (b) edge [bend left=5pt] node [below,csemi,yshift=2pt,pos=0.6] {$s$} (a)
        (a) edge [bend left=5pt] node [above,csemi] {$s$} (a1)
        (a1) edge [bend left=5pt] node [below,csemi] {$s$} (a)
        (a) edge [bend left=5pt] node [above,csemi] {$s$} (a2)
        (a2) edge [bend left=5pt] node [below,csemi] {$s$} (a)
        (b) edge [bend left=5pt] node [above,csemi] {$s$} (a1)
        (a1) edge [bend left=5pt] node [below,csemi] {$s$} (b)
        (b) edge [bend left=5pt] node [pos=0.7,above,csemi] {$s$} (a2)
        (a2) edge [bend left=5pt] node [pos=0.7,below,csemi] {$s$} (b)
        (a1) edge [bend left=5pt] node [pos=0.63,right,csemi,inner sep=2pt] {$s$} (a2)
        (a2) edge [bend left=5pt] node [pos=0.63,left,csemi,inner sep=2pt] {$s$} (a1)
        ;
\path [->,>=stealth',shorten <=1pt,shorten >=1pt,every node/.style={font=\scriptsize}]
        (a) edge [loop left] node [left,csemi] {$s$} (a)
        (b) edge [in=-15,out=20,looseness=7,shorten <=1pt,shorten >=1pt] node [right,csemi,inner sep=1pt] {$s$} (b)
        (a1) edge [loop right] node [right,csemi] {$s$} (a1)
        (a2) edge [loop right] node [right,csemi] {$s$} (a2)
        ;
\node [font=\scriptsize,colu,anchor=south east,inner sep=0pt] at (a.120) {$t$};
\node [font=\scriptsize,colu,anchor=south east,inner sep=0pt] at (b.120) {$t$};
\node [font=\scriptsize,colu,anchor=south east,inner sep=0pt] at (a1.120) {$t$};
\node [font=\scriptsize,colu,anchor=north east,inner sep=1pt] at (a2.190) {$t$};
\end{tikzpicture}
\caption{Illustration of the $K$-interpretations $\pi_2$ and $\pi'_2$ (in the proof of \cref{thmDistanceCounterexample}).}
\label{figDistanceCounterexample}
\end{figure}

Recall that basic local sentences are of the form $\eta = \Esc r {y_1,\dots,y_m} \bigland_i \phi \loc r(y_i)$ or $\eta' = \Asc r {y_1,\dots,y_m} \biglor_i \phi \loc r(y_i)$.
We classify these sentences as follows.
\begin{bracketise}
    \item $r \ge 1$ and $m \ge 2$: due to the values of the distance formulae in the scattered quantifiers, we have $\pi_n \ext \eta = \pi_n' \ext \eta = 0$ and $\pi_n \ext {\eta'}, \pi_n' \ext {\eta'} \ge s$.
    %, as all elements have distance $\le 1$.
    
    \item $r = 0$: local formulae $\phi \loc {0}(x)$ of radius $0$ can only check atoms $Ux$ or $Exx$, whereas the only difference between $\pi_n$ and $\pi_n'$ is the value of the atom $Evw$. Hence both interpretations assign the same values, i.e., $\pi_n \ext \eta = \pi_n' \ext \eta$ and $\pi_n \ext {\eta'} = \pi_n' \ext {\eta'}$.
    
    \item $r \ge 1$ and $m = 1$ (discussed below).
\end{bracketise}

Fix some $n$ for the moment.
Since we assume that $\beta \keq \gamma_1 \lor \dots \lor \gamma_k$, we have $\pi_n\ext{\gamma_i} = 0$ for all $i$.
Moreover, there must be some $i$ with $\pi_n' \ext {\gamma_i} = t$.
As $\gamma_i$ is a conjunction of basic local sentences, it must contain a basic local sentence $\psi$ such that $\pi_n \ext {\psi} = 0$ and $\pi_n' \ext {\psi} \ge t$.
This sentence can only be of type \enumref{3}.

This holds for each $n$.
Since $\gamma$ is finite, there must thus be a basic local sentence $\psi$ of type \enumref{3} such that $\pi_n \ext {\psi} = 0$ and $\pi_n' \ext {\psi} \ge t$ for infinitely many $n$.
We prove that no such sentence exists.

\begin{localclaim}
\label{claim1}
Let $\psi$ be of type \enumref{3}, i.e., $\psi = \E x \, \phi \loc r (x)$ or $\psi = \A x \, \phi \loc r(x)$ with $r \ge 1$, so
that $\pi_n \ext \psi = 0$ for infinitely many $n$.
Then there is $n_0 \in \N$ such that $\pi_n' \ext {\psi} \le s$ for all $n \ge n_0$.
\end{localclaim}

Intuitively, the only way that $\psi$ can evaluate to $0$ in $\pi_n$, but to a positive value in $\pi_n'$, is by using the literal $\neg Evw$.
But this requires to quantify $w$ local around $v$ (or vice versa), which forces the value of $\psi$ to be $\le s$.
Indeed, a universal quantifier would range also over $\neg Evv$ and thus evaluate to $0$; an existential quantifier would include the distance formula $d(v,w) \le r$ which evaluates to $s$.

The formal proof of this intuition is rather technical.
We first distil $\phi \loc r(x)$ down to a simpler formula still witnessing $\pi_n' \ext \psi \le s$ (\cref{claim2} below), and then use this to prove \cref{claim1}.
We write $\ta = (a,\dots,a)$ for a tuple of matching arity containing only the element $a$.

\begin{localclaim}
\label{claim2}
Let $a \in \{v,w,v_1\}$.
Let $\phi \loc r (x)$ be a local formula with $\pi_n \ext {\phi \loc r(a)} = 0$ for infinitely many $n$.
There is a quantifier-free formula $\phi'(x,\ty,\tz)$, $n_0 \in \N$ and $\tc \subseteq A_{n_0}$ s.t.,
\begin{itemize}
\item at most one of $v,w$ occurs in $a,\tc$,
\item $\pi_n \ext {\phi'(a,\ta,\tc)} = 0$ for all $n \ge n_0$,
\item if $\pi_n' \ext {\phi'(a,\ta,\tc)} \le s$ then also $\pi_n' \ext {\phi \loc r(a)} \le s$, for all $n \ge n_0$.
\end{itemize}
\end{localclaim}

\begin{claimproof}[Proof of \cref{claim2}]
We construct $\phi'$, $n_0$ and $\tc$ from $\phi \loc r(x)$ by eliminating quantifiers.
\begin{enumise}
\item
First obtain $\phi_1 \loc r(x, \ty)$ from $\phi$ by turning all existentially quantified variables into free variables.
That is, we inductively replace all relativised existential quantifiers $\E y (d(x,y) \le r \land \theta(x,y))$ by just $\theta(x,y)$, with $y$ becoming a free variable.

\smallskip
Then $\pi_n \ext {\phi \loc r(a)} = 0$ implies $\pi_n\ext{\phi_1 \loc r(a, \ta)} = 0$.
Further, if $\pi_n'\ext{\phi_1(a,\ta)} \le s$, then $\pi_n' \ext{\phi \loc r(a)} \le s$, for all $n$.
To see this, note that a relativised quantifier $\E y (d(x,y) \le r \land \theta(x,y))$ in $\phi \loc r(a)$ evaluates to a value $\le s$ due to the distance formula, unless we instantiate $y$ by $a$ (then $d(a,a) \le r$ evaluates to $1$).
But if we have $\pi_n'\ext{\phi_1(a,\ta)} \le s$, then also this instantiation cannot lead to larger values, and hence a simple induction gives $\pi_n' \ext {\phi \loc r(a)} \le s$ as claimed.

\item
We next turn all relativised universal quantifiers into (non-relativised) universal quantifiers.
That is, $\A z (d(x,z) > r \lor \theta(x,\ty,z))$ becomes $\A z \, \theta(x,\ty,z)$.
Note that this does not change the value of the formula in $\pi_n'$ and $\pi_n$, as $d(x,z) > r$ always evaluates to $0$.

\smallskip
By applying prenex normal form, we obtain a formula of the form $\A z_1 \dots \A z_m \, \phi_2(x,\ty,\tz)$, where $\phi_2$ is quantifier free.
As our transformations did not affect the evaluation, we can find an $n_0 > m+1$ so that $\pi_{n_0} \ext {\A z_1 \dots \A z_m\,\phi_2(a,\ta,\tz)} = 0$.
Let further $\tc \subseteq A_{n_0}$ be a witnessing instantiation for $\tz$ such that $\pi_{n_0} \ext {\phi_2(a,\ta,\tc)} = 0$.
Since $\phi_2$ is quantifier-free, this implies $\pi_n \ext {\phi_2(a,\ta,\tc)} = 0$ for all $n \ge n_0$.
Notice that we can choose $\tc$ so that at most one of $v,w$ occurs in $\phi_2(a,\ta,\tc)$.
Indeed, $\pi_{n_0}$ is completely symmetric (every permutation of $A_{n_0}$ is an automorphism) and we only fix a single element $a$.
So if $a \neq v$ and the witness contains $v$ (and analogously for $w$ if $a \neq w$), we can replace $v$ by an unused element among $v_1,\dots,v_{n_0}$ without affecting the value $\pi_{n_0} \ext {\phi_2(a,\ta,\tc)}$.

\smallskip
For the condition in $\pi_n'$, assume that $\pi_n' \ext {\phi_2(a,\ta,\tc)} \le s$ for some $n \ge n_0$.
Then also $\pi_n' \ext {\phi_1 \loc r(a,\ta)} = \pi_n' \ext {\A z_1 \dots \A z_m \, \phi_2(a,\ta,\tz)} \le s$, and thus $\pi_n' \ext {\phi \loc r(a)} \le s$ (step 1).
\claimqedhere
\end{enumise}
\end{claimproof}

\medskip
It remains to prove \cref{claim1}.
First consider the case $\psi = \E x \, \phi \loc r(x)$.
Recall that we assume $\pi_n \ext \psi = 0$ for infinitely many all $n$, hence $\pi_n \ext {\phi \loc r(a)} = 0$ for all $a \in \{v,w,v_1\}$.
Obtain $\phi'(x,\ty,\tz)$, $n_0'$ and $\tc$ by \cref{claim2}.
Then $\pi_n \ext {\phi'(a,\ta,\tc)} = 0$ for all $n \ge n_0$ and all $a \in \{v,w,v_1\}$.
Notice that $\phi'(a,\ta,\tc)$ does not contain the literals $Evw$ and $\neg Evw$, since at most one of $v,w$ can appear.
But then each occurring literal is mapped to the same value by $\pi_n'$ and $\pi_n$, so  $\pi_n' \ext {\phi'(a,\ta,\tc)} = 0$ for $n \ge n_0$ as well.
By \cref{claim2}, this implies $\pi_n' \ext {\phi\loc r(a)} \le s$ for all $n \ge n_0$.
Finally, notice that each $\pi_n'$ is symmetric with respect to the nodes $v_1,\dots,v_n$ (every permutation of $v_1,\dots,v_n$ is an automorphism of $\pi_n'$), and thus $\pi_n' \ext {\phi \loc r(v_i)} = \pi_n' \ext {\phi\loc r(v_1)} \le s$ for all $i \le n$.
Altogether, we have shown $\pi_n' \ext {\E x \, \phi\loc r(x)} \le s$ for all $n \ge n_0$.

Now consider $\psi = \A x \, \phi \loc r(x)$.
Since $\pi_n \ext \psi = 0$ for infinitely many $n$, there is an element $a_n \in A_n$ for each $n$ such that $\pi_n \ext {\phi \loc r(a_n)} = 0$.
However, since $\pi_n$ is completely symmetric, we can always choose $a_n = v$, so $\pi_n \ext {\phi \loc r(v)} = 0$ for infinitely many $n$.
By applying \cref{claim2} as above, we get $\pi_n' \ext {\phi \loc r(v)} \le s$ and thus $\pi_n' \ext {\A x \, \phi \loc r(x)} \le s$, for all $n \ge n_0$.
\end{proof}

\end{document}